\documentclass[twocolumn,secnumarabic,nobibnotes,superscriptaddress,aps,prx]{revtex4-2}
\usepackage{color, xcolor, colortbl}
\usepackage{graphicx}
\usepackage{amsthm,amsmath,amssymb,amsfonts}
\usepackage{dcolumn}% Align table columns on decimal point
\usepackage{url}
\usepackage{epstopdf}
\usepackage{enumitem}
\usepackage{algorithm}
\usepackage{algpseudocode}
\usepackage{bm}
\usepackage[caption=false]{subfig}
\usepackage{appendix}
\usepackage{multirow}
\usepackage{braket}
\usepackage[english]{babel}
\usepackage{hyperref}
\usepackage[capitalize]{cleveref}
\usepackage{xpatch}
\usepackage{tikz}
\usepackage{adjustbox}
\usepackage{xspace}
\usetikzlibrary{quantikz}

\renewcommand{\Re}{\operatorname{Re}}
\renewcommand{\Im}{\operatorname{Im}}

\newcommand{\I}{\mathrm{i}}

\newcommand{\mc}[1]{\mathcal{#1}}

\newcommand{\wt}[1]{\widetilde{#1}}

\newcommand{\abs}[1]{\left\lvert#1\right\rvert}
\newcommand{\norm}[1]{\left\lVert#1\right\rVert}

\newcommand{\Or}{\mathcal{O}}

\newcommand{\RR}{\mathbb{R}}

\newtheorem{thm}{\protect\theoremname}
\newtheorem{lem}[thm]{\protect\lemmaname}

\newtheorem{defn}[thm]{\protect\definitionname}

\providecommand{\definitionname}{Definition}
\providecommand{\assumptionname}{Assumption}
\providecommand{\corollaryname}{Corollary}
\providecommand{\lemmaname}{Lemma}
\providecommand{\propositionname}{Proposition}
\providecommand{\remarkname}{Remark}
\providecommand{\theoremname}{Theorem}
\usepackage{bbm}



\usetikzlibrary{fit}
\tikzset{%
  highlight/.style={rectangle,rounded corners,fill=blue!15,draw,fill opacity=0.3,thick,inner sep=0pt}
}

\newcommand{\PP}{\mathbb{P}}
\newcommand{\squares}{\Lambda}
\theoremstyle{definition}

\newcommand{\rads}{r}

\newcommand{\UMMath}{Department of Mathematics, University of Michigan, Ann Arbor, MI 48109, USA}
\newcommand{\UMEECS}{Department of Electrical Engineering and Computer Science, University of Michigan, Ann Arbor, MI 48109, USA}
\newcommand{\UMDCS}{Department of Computer Science, University of Maryland, College Park, MD 20742, USA}
\newcommand{\QUICS}{Joint Center for Quantum Information and Computer Science, University of Maryland, College Park, MD 20742, USA}

\begin{document}

\title{Programmable Signal Design for Quantum Phase Estimation via Quantum Signal Processing}

\author{Zikang Jia}
\thanks{These authors contributed equally.}
\affiliation{\UMMath}

\author{Suying Liu}
\thanks{These authors contributed equally.}
\affiliation{\UMDCS}
\affiliation{\QUICS}

\author{Yulong Dong}
\email[Electronic address: ]{dongyl@umich.edu}
\affiliation{\UMEECS}

\date{\today}

\begin{abstract}
Quantum phase estimation is a central primitive in quantum algorithms and sensing, where performance is governed by the sensitivity of measurement signals to the target parameter. While existing methods have developed increasingly sophisticated inference and adaptive design strategies, the signal family used for phase learning is often largely pre-specified. Here we propose a programmable signal design framework for quantum phase estimation based on quantum signal processing, which enables the measurement signal to be tailored to the current uncertainty region. We cast phase estimation as a max-min optimization problem over admissible signals and introduce a sensitivity efficiency parameter that quantifies information gain per query depth. The resulting iterative algorithm combines optimized quantum signal transformations with structured classical inference, retaining Heisenberg-limited scaling while improving sensitivity efficiency and practical resource prefactors. Numerical results show reduced estimation variance compared with standard protocols such as robust phase estimation. Our framework also extends to Hamiltonian eigenvalue estimation in higher dimensions and establishes a quantum-classical co-design paradigm through programmable signal shaping.
\end{abstract}

\maketitle

\section{Introduction}

Quantum phase estimation is a central primitive in quantum algorithms, quantum simulation, and quantum sensing \cite{martyn_grand_2021,aspuru2005simulated,fomichev_initial_2024,shulman_suppressing_2014,obrien_error_2021}. It aims to learn an unknown parameter $\theta$ associated with a queryable unitary from measurement outcomes of parameter-dependent quantum circuits. In many settings, these outcomes define probabilities or expectation values that can be viewed as signals of $\theta$. The performance of a phase estimation protocol is governed by the local behavior of these signals, especially their derivatives near the target parameter. This local sensitivity is closely related to the Fisher information and determines the total quantum resources, such as the total query count, required to reach a target precision \cite{giovannetti_quantum_2006,giovannetti_advances_2011,gorecki_ensuremathpi-corrected_2020}.

A broad class of phase estimation methods focuses on designing efficient inference procedures within a restricted family of measurement architectures, typically based on Hadamard-test-type circuits \cite{knill_optimal_2007,berry_how_2009,kimmel_robust_2015,ding_even_2023,ni_low-depth_2023,dong_optimal_2025}. Recent works have also explored adaptive experiment design \cite{wiebe_efficient_2016,yamamoto_demonstrating_2024,lumino_experimental_2018}, filtering strategies \cite{dong_ground-state_2022,dong_multi-level_2024,karacan_enhancing_2025}, and error-corrected metrology schemes \cite{yamamoto_quantum_2025,marrero_encoded_2026} in related spectral-estimation and metrology settings. Related subspace-based approaches construct a trial subspace from time-evolved states and extract spectral information through a projected classical problem, with the main flexibility lying in subspace construction and post-processing \cite{shen_efficient_2026,cortes_quantum_2022,stair_stochastic_2023}. Nevertheless, in most existing phase estimation approaches, the signal used for inference still comes from a largely pre-specified family, such as complementary sine- and cosine-type responses, rather than being programmably synthesized to match the local uncertainty region of the target parameter. As a result, prior information is used mainly to choose experiment settings or inference procedures, while the local signal shape itself remains only weakly tunable.

This observation motivates a different perspective on phase estimation. Classical post-processing determines how information is extracted from measurement data, whereas quantum circuits determine how that information is encoded into the data in the first place. From this viewpoint, phase estimation should be treated not only as an inference problem, but also as a signal design problem. The missing degree of freedom in standard approaches is the ability to tailor the signal family itself to the current uncertainty region of $\theta$. This suggests a co-design paradigm in which both signal generation and inference are adapted to the estimation task.

In this Letter, we realize this idea using quantum signal processing (QSP). QSP provides a systematic framework for programming the signal family through phase factors \cite{low_methodology_2016,low_optimal_2017,gilyen_quantum_2019}, far beyond the fixed single-mode trigonometric responses of standard constructions. By using QSP to reshape the local response of the measurement signal near the target phase, we develop a co-designed phase estimation framework that combines programmable signal transformations with tailored classical post-processing. We term the resulting framework \emph{QSP-based phase estimation} (QSP-PE). We establish convergence guarantees and quantify the resulting improvement in sensitivity efficiency. Numerical results confirm favorable performance compared with standard phase estimation protocols such as robust phase estimation (RPE) \cite{belliardo_achieving_2020,ni_low-depth_2023}. The framework also extends naturally to Hamiltonian eigenvalue estimation through quantum eigenvalue transformation circuits \cite{dong_ground-state_2022}. A schematic overview of the approach is shown in Fig.~\ref{fig:main}. Overall, our results show that quantum algorithmic tools can be used not only for computation, but also for engineering phase-estimation protocols through programmable signal design.

\begin{figure}[htbp]
\centering
\includegraphics[width=\linewidth]{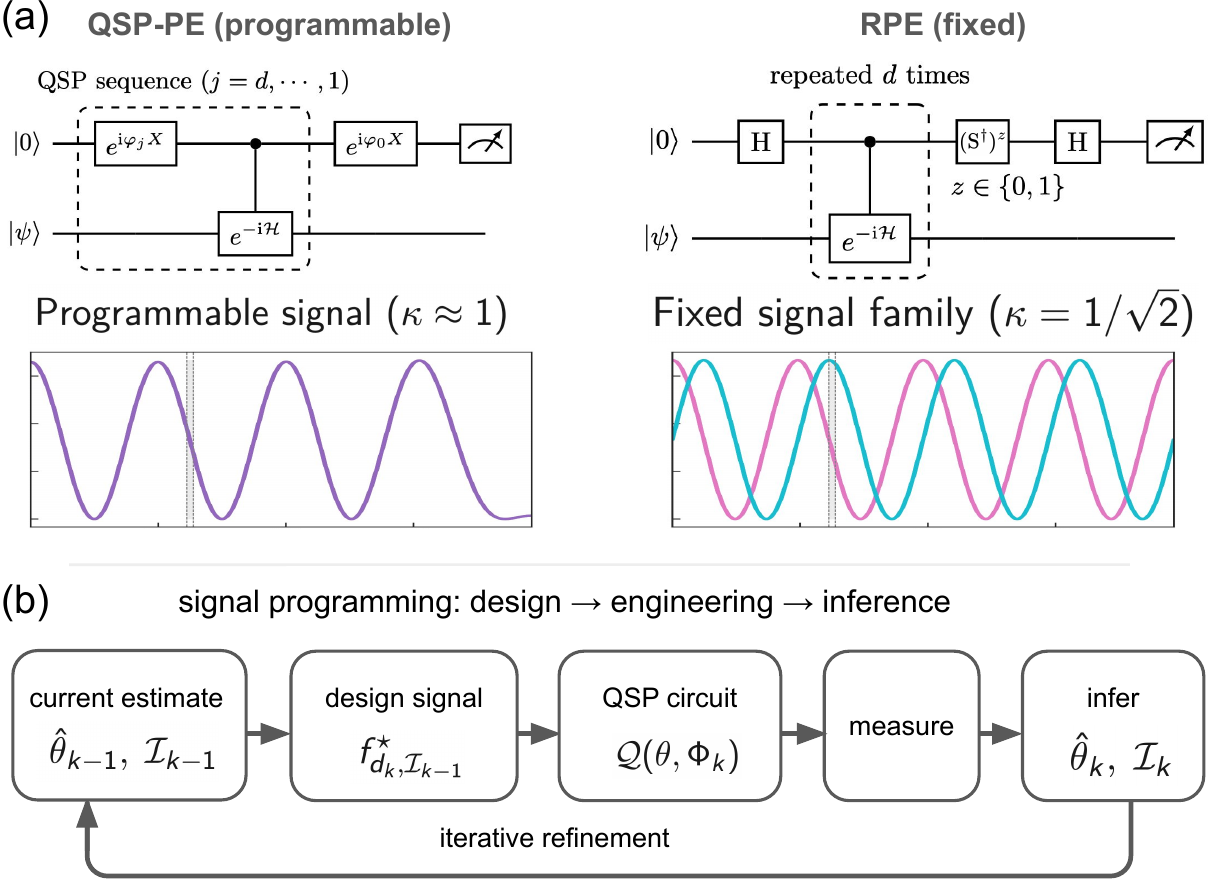}
\caption{Signal programming framework for quantum phase estimation. (a) Comparison between standard phase estimation based on fixed signal families (e.g., robust phase estimation, RPE) and the proposed QSP-based approach. Conventional methods rely on fixed single-mode trigonometric signals, which limit the achievable sensitivity efficiency to $\kappa = 1/\sqrt{2}$. In contrast, QSP enables programmable signal transformations, allowing the signal to be tailored to a local confidence interval and achieving near-optimal sensitivity $\kappa \approx 1$. (b) End-to-end signal programming framework. Starting from a current estimate, the method designs an optimal signal by solving a max-min problem, implements it using a QSP circuit, and performs measurement and classical inference to obtain an updated estimate. This closed-loop procedure iteratively refines the estimate with improved sensitivity.} 
\label{fig:main}
\end{figure}

\section{Signal formulation of phase estimation}

To formalize the above idea, we cast phase estimation as a signal design problem. Let $\theta$ denote the target phase parameter, and let $f(\theta)$ denote the measurement signal produced by a parameter-dependent quantum circuit. In our setting, the signal family is generated by quantum signal processing (QSP), a quantum algorithmic framework for implementing polynomial transformations of the eigenvalues of a unitary operator \cite{low_optimal_2017,gilyen_quantum_2019,dong_ground-state_2022,motlagh_generalized_2024,martyn_grand_2021}. By interleaving queries to the controlled unitary with single-qubit ancilla rotations, QSP produces a broad class of functions of the phase parameter.

A key feature of QSP is its qubitization structure, which yields an effective $\mathrm{SU}(2)$ description of the eigenvalue transformation on a two-dimensional invariant subspace for each eigenvalue \cite{low_hamiltonian_2019,gilyen_quantum_2019}. We therefore first formulate the design problem in this $\mathrm{SU}(2)$ setting, and return to the higher-dimensional lift later. Concretely, given phase factors $\Phi=(\varphi_0,\ldots,\varphi_d)$, the associated QSP sequence is
\begin{equation}
    \mathcal{Q}(\theta,\Phi) = e^{\I \varphi_0 X} \prod_{j=1}^d \left( e^{\I \theta Z} e^{\I \varphi_j X} \right).
\end{equation}
For a circuit of depth $d$, the resulting signal takes the form
\begin{equation}
    f(\theta)=g^2(\theta):=\abs{\braket{0|\mathcal{Q}(\theta,\Phi)|0}}^2,\qquad g\in\mathcal{G}_d,
\end{equation}
where $\mathcal{G}_d$ denotes the class of bounded trigonometric polynomials accessible to QSP at depth $d$. We write $\mathcal{F}_d:=\left\{f:f(\theta)=g^2(\theta),\,g\in\mathcal{G}_d\right\}$ for the corresponding admissible signal class.

The performance of the estimator is controlled by the local behavior of the signal on a confidence interval $\mathcal{I}$ around the target parameter. Such an interval arises naturally in iterative estimation procedures, where each stage produces an estimate with finite uncertainty and thereby defines a local region containing the true parameter with high probability. Intuitively, by local error propagation, the induced parameter uncertainty scales as $\Delta\theta \sim \Delta f / |f'(\theta)|$ on $\mathcal{I}$, so a larger local slope means that the same uncertainty in the measured signal translates into a smaller uncertainty in $\theta$.

This motivates the following design principle: for a given query depth $d$ and uncertainty interval $\mathcal{I}$, choose the signal within the admissible class to maximize the worst-case sensitivity,
\begin{equation}\label{eqn:max-min-design-signal}
    L(f, \mc{I}) := \min_{\theta \in \mathcal{I}} |f'(\theta)|,\quad f^\star_{d,\mathcal{I}}
    = \arg\max_{f \in \mathcal{F}_d}
      L(f, \mc{I}).
\end{equation}
This formulation turns phase estimation into a max-min signal design problem.

To quantify the sensitivity achieved per unit quantum resource, we define the sensitivity efficiency
\begin{equation}
    \kappa(f,\mathcal{I}) := \frac{L(f, \mc{I})}{d},
\end{equation}
where $d$ denotes the query depth. Since the derivative of a degree-$d$ trigonometric polynomial is fundamentally bounded by $d$ through Bernstein's inequality, $\kappa(f,\mathcal{I})$ provides a natural normalized measure of how efficiently a protocol converts coherent query depth into useful local sensitivity over the target interval. Larger $\kappa$ therefore leads to smaller estimation variance at fixed shot number and query depth.

Standard phase-estimation protocols such as Robust Phase Estimation (RPE) \cite{belliardo_achieving_2020,ni_low-depth_2023} are based on fixed trigonometric signals $\left\{(1+\cos(2d\theta))/2,\ (1+\sin(2d\theta))/2\right\}$. These signals correspond to two complementary phase transformations, which together avoid blind spots in sensitivity when either component vanishes. However, this also distributes the available sensitivity across two separate channels, so that each individual signal carries only a fraction of the maximal slope. As a result, the effective sensitivity efficiency is $\kappa_{\mathrm{RPE}} = 1/\sqrt{2}$. By contrast, the versatility of QSP-based signals enables substantially larger uniform slope over the same interval, leading to higher sensitivity efficiency and improved constant factors in estimation precision.

\section{Algorithm and Performance Guarantee}

We develop an iterative phase estimation procedure based on progressive refinement of both the estimator and its associated confidence interval. At each stage, the current interval quantifies the remaining uncertainty of the parameter. Given this interval, we solve the max-min signal design problem in \cref{eqn:max-min-design-signal} to construct an optimized signal beyond standard single-mode trigonometric functions. The resulting signal achieves uniformly high sensitivity over the interval and can be implemented by a QSP circuit, with the quantum enhancement arising from the improved sensitivity efficiency $\kappa$ compared to standard phase-estimation signals. Measurement outcomes are then used to update the estimate and shrink the confidence interval. 

In some phase-estimation methods, the classical post-processing step is formulated as loss minimization and can be hindered by nonconvexity and spurious local minima \cite{neill_accurately_2021,arute_observation_2020}. In contrast, in our QSP-based method, the designed signal is locally monotone on the confidence interval, so the inference problem in post-processing is no longer a generic optimization problem but reduces exactly to one-dimensional root finding. This structured local geometry makes the classical post-processing provably simple and efficient, enabling stable updates via standard methods such as bisection. Repeating this procedure progressively refines the interval and improves the estimation accuracy.

We summarize the iterative QSP-based estimation procedure in the following pseudo-code.

\begin{algorithm}[H]
\caption{QSP-based Phase Estimation (QSP-PE)}
\begin{algorithmic}
\State Initialize a confidence interval $\mathcal{I}_0$ centered at $\hat \theta_0$
\For{$k = 1,2,\dots,K$}
    \State Design a signal $f_k$ by solving the max--min problem over $\mathcal{I}_{k-1}$
    \State Implement $f_k$ using a QSP circuit with query depth $d_k$
    \State Collect $m_k$ measurement samples
    \State Update the estimate $\hat{\theta}_k$ via classical post-processing
    \State Shrink the confidence interval $\mathcal{I}_k$
\EndFor
\State \Return $\hat{\theta}_K$ and $\mc{I}_K$
\end{algorithmic}
\end{algorithm}

The pseudo-code highlights the iterative structure of the procedure, where quantum signal design and classical post-processing are combined to progressively refine the estimate. The performance of this procedure can be characterized in terms of the sensitivity efficiency $\kappa$, which quantifies how effectively quantum circuit depth translates into useful information about the parameter.

\begin{thm}\label{thm:QSP-PE-performance-guarantee}
For any target precision $\epsilon > 0$ and failure probability $\delta \in (0,1)$, the proposed iterative QSP-based phase estimation procedure outputs an estimator $\hat{\theta}$ satisfying $\mathbb{P}(|\hat{\theta}-\theta| \le \epsilon) \ge 1-\delta$.

Moreover, if the employed signal family has sensitivity efficiency $\kappa$, then it suffices to take
$m = \mathcal{O}\!\left(\kappa^{-2}\log\frac{\log(1/\epsilon)}{\delta}\right)$
measurements per stage. Under this condition, the number of refinement stages is $K=\mathcal{O}(\log(1/\epsilon))$, the maximal query depth scales as $\mathcal{O}(1/\epsilon)$, the total query complexity is $\widetilde{\mathcal{O}}\!\left(\frac{1}{\kappa^2\epsilon}\log\frac{1}{\delta}\right)$, and the total classical floating-point complexity is $\mathcal{O}(1/\epsilon)$ in post-processing.
\end{thm}

The proof of the theorem is given in \cref{app:sec:resource_analysis}. This theorem shows that the complexity of our algorithm achieves Heisenberg-limited scaling. Moreover, by optimizing the sensitivity efficiency $\kappa$, the overall cost can be further reduced. Owing to the local monotonicity of the signal, the post-processing step remains computationally efficient and avoids the need for black-box optimization.

When the number of measurement shots exceeds the minimal requirement in the theorem, the confidence interval can be further refined. In the baseline setting corresponding to the minimal sample size, the confidence radius is $r_k = 1/(4d_{k+1})$. For a larger number of measurements, the improved statistical concentration leads to a more accurate characterization of the signal, allowing the interval to shrink proportionally. More precisely, if $m_{\mathrm{actual}} = \zeta^2 \, m_{\mathrm{minimal}}$ for some $\zeta \ge 1$, then the confidence radius can be reduced to $r_k = 1/(4 \zeta d_{k+1})$. This refinement leads to a narrower working interval for signal design, which in turn enables higher achievable sensitivity.

\section{Solving the Max-Min Signal Design Problem}

Directly solving the optimal design problem in \cref{eqn:max-min-design-signal} is numerically challenging due to the nonconvexity induced by the quadratic relation $f(\theta)=g^2(\theta)$, where $g$ is a QSP-admissible polynomial. To address this, we introduce a relaxed optimization problem that is much easier to solve in practice. The key observation is that $f'(\theta)=2g(\theta)g'(\theta)$. Therefore, if both $|g(\theta)|$ and $|g'(\theta)|$ are uniformly bounded away from zero on the confidence interval, then the resulting signal $f(\theta)=g^2(\theta)$ also has large sensitivity on that interval.

Motivated by this observation, we replace the original max-min problem by a coupled optimization over the unsquared transformation function $g \in \mathcal{G}_d$. For a given lower-bound parameter $\alpha$ on the magnitude of $g$, we consider the relaxed problem, where an auxiliary variable $\beta$ captures the minimal derivative magnitude on teh confidence interval:
\begin{equation}
\begin{array}{ll}
    \mathrm{maximize} \;\; &  \beta\\
    \mathrm{subject\ to} \;\; &
    g \in \mathcal{G}_d,\\
    & |g(\theta)| \ge \alpha,\ 
    |g'(\theta)| \ge \beta,
    \; \theta \in \mathcal{I}.
\end{array}
\end{equation}
This relaxed problem can be solved efficiently by linear programming after discretizing the interval constraints \cite{boyd2004convex,gurobi}, and in practice provides high-quality signals for the original design objective.

\section{Numerical Results}

We now evaluate the performance of the proposed approach and compare it with standard phase-estimation methods. The key quantity is the sensitivity efficiency $\kappa$, which determines how effectively query depth translates into useful information about the parameter.

\begin{figure}[htbp]
\centering
\includegraphics[width=\linewidth]{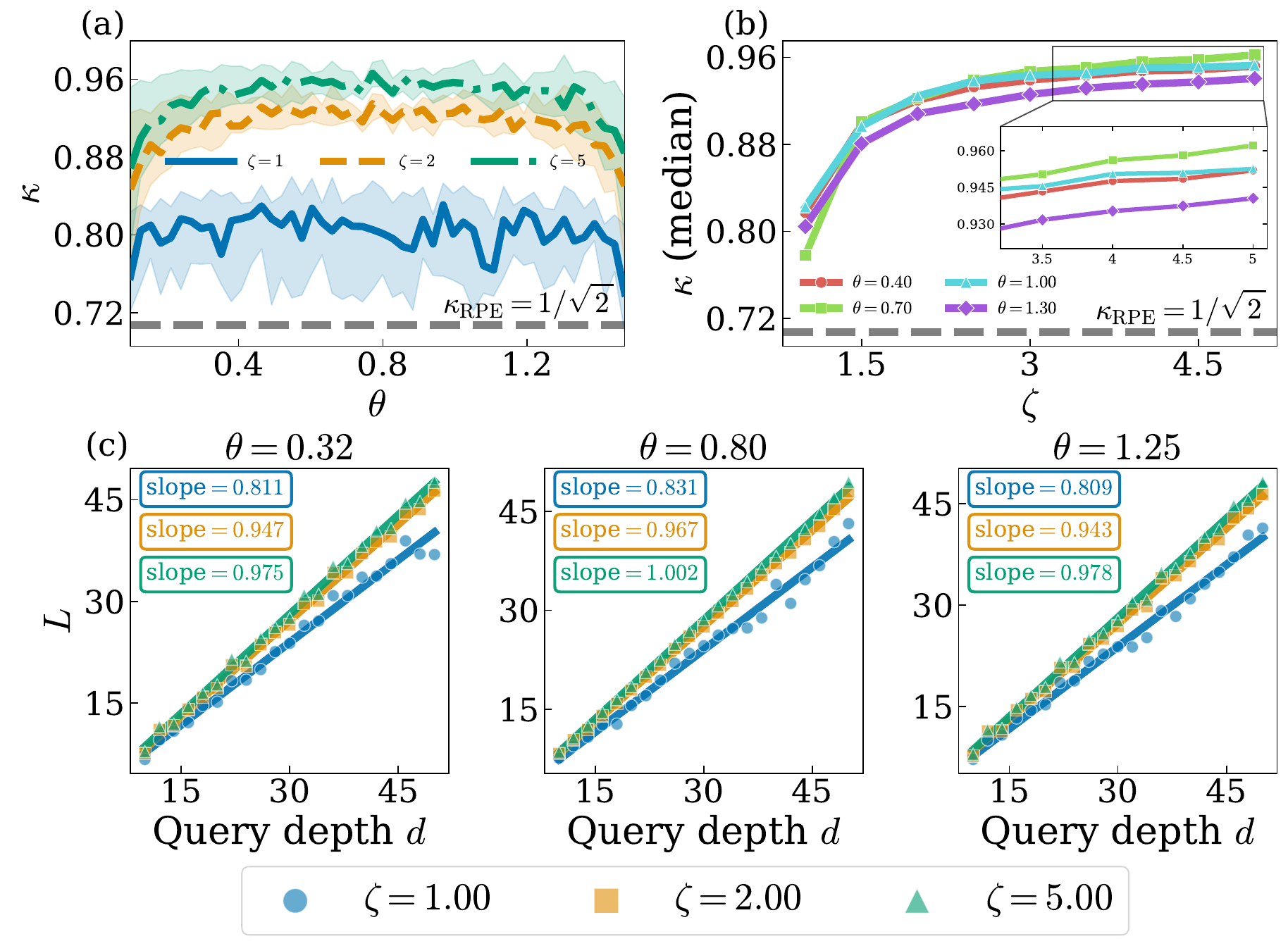}
\caption{(a) Sensitivity efficiency $\kappa$ as a function of target phase $\theta$. $\kappa$ improves uniformly as $\zeta$ increases across all tested values of $\theta$. The shaded region depicts the 10\%-90\% percentile range of $\kappa$. The sensitivity efficiency of the optimal design consistently exceeds that of RPE, indicated by the dashed line at $\kappa_{\text{RPE}} = 1/\sqrt{2}$. (b) Sensitivity efficiency $\kappa$ as a function of shrinkage factor $\zeta$. $\kappa$ increases monotonically with $\zeta$, and this improvement is uniform across all target phases $\theta_0$. (c) Minimum derivative lower bound $L$ as a function of query depth $d$, evaluated at three target phases $\theta \in \{0.32, 0.80, 1.25\}$ and shrinkage factors $\zeta \in \{1, 2, 5\}$. For each combination of $\theta$ and $\zeta$, $L$ is plotted against $d$ and fit with a linear regression (solid lines). The fitted slopes increase toward $1$ as $\zeta$ grows.} 
\label{fig:kappa-theta-zeta-L-d-regression}
\end{figure}

\cref{fig:kappa-theta-zeta-L-d-regression}(a) shows that the optimized QSP-based signals achieve significantly higher sensitivity efficiency compared to the fixed trigonometric signals used in robust phase estimation (RPE). The improvement is consistent across different values of the target phase, indicating that programmable signal design enables a systematic enhancement beyond standard Hadamard-test-based circuit constructions. \cref{fig:kappa-theta-zeta-L-d-regression}(b) further demonstrates that increasing the measurement budget, equivalently increasing $\zeta$ parameter, leads to higher achievable sensitivity. This is because a larger number of measurements allows for a more refined confidence interval, enabling signal design over a narrower domain.

The scaling study is shown in \cref{fig:kappa-theta-zeta-L-d-regression}(c), where the worst-case derivative $L$ is plotted against the query depth $d$. Across different target phases, the numerical results exhibit a linear relation, with slope agreeing with the sensitivity efficiency $\kappa$. This confirms that the achievable sensitivity is well characterized by the scaling relation $L \approx \kappa d$, and that $\kappa$ remains approximately consistent over depths. The observed linearity across different values of the target phase also indicates that the method is broadly applicable rather than restricted to isolated parameter regimes. Moreover, as $\zeta$ increases, the slope becomes larger and approaches $1$, showing that a narrower confidence interval enables a systematic improvement in $\kappa$.

\begin{figure}[htbp]
\centering
\includegraphics[width=\linewidth]{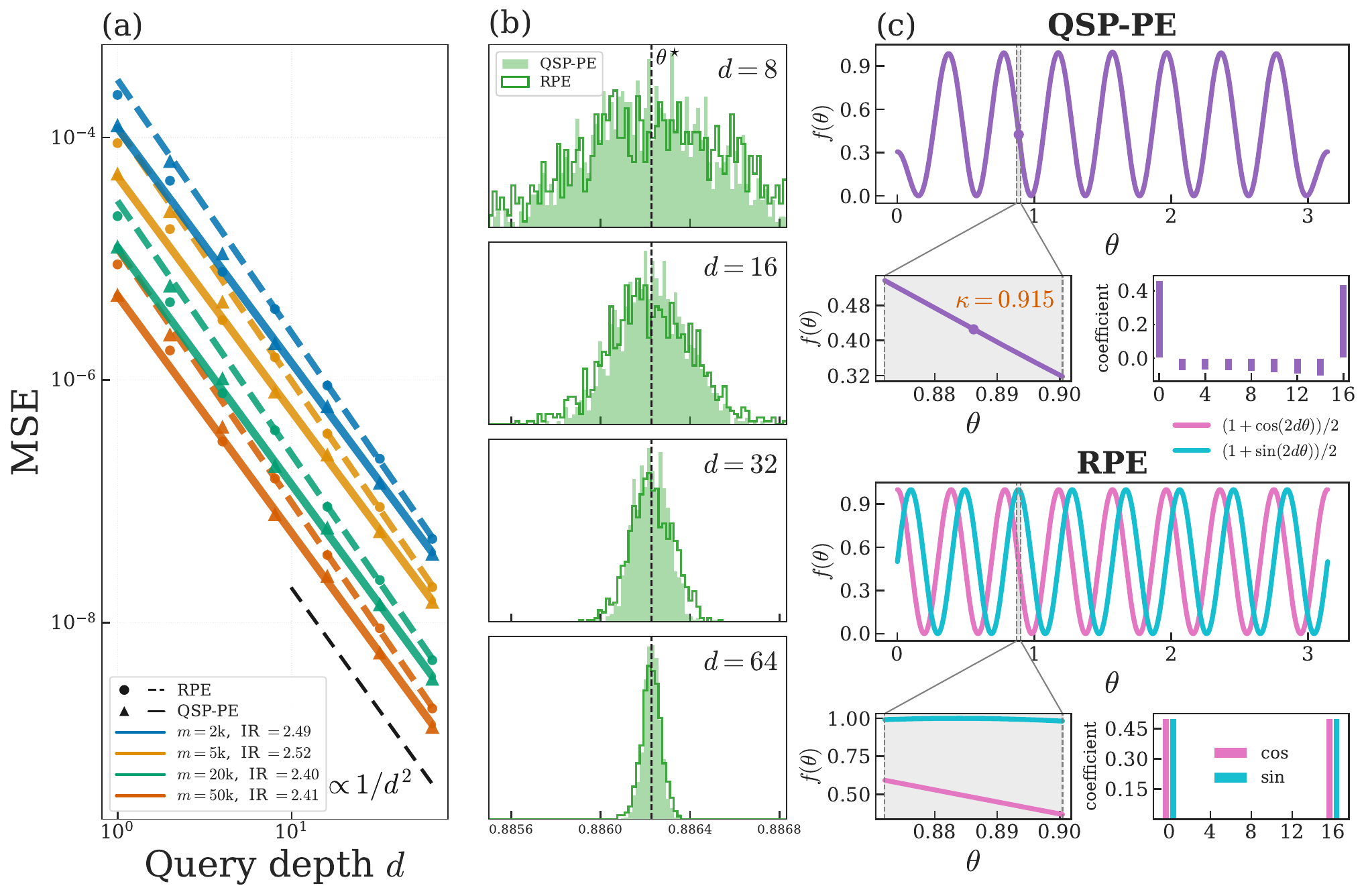}
\caption{(a) Numerical comparison of the MSE of the $\hat{\theta}$ estimator for QSP-PE and RPE across query depth $d$, where the target phase is $\theta^{\star} = \sqrt{\pi}/2$. Linear regression is applied to the empirical MSE at depths $d = 8, 16, 32, 64$, and the improvement ratio (IR) is defined as the MSE ratio of RPE to QSP-PE derived from the regression intercepts. Across measurement budgets from $2{,}000$ to $50{,}000$ shots, the average IR is ${\sim}2.46$, consistent with the theoretical prediction. (b) Distributions of the $\hat{\theta}$ estimator from QSP-PE and RPE at query depths $d = 8, 16, 32, 64$ for $50{,}000$ shots with $\zeta=\sqrt{10}$. At each depth, QSP-PE estimates are more tightly concentrated around $\theta^{\star}$, directly accounting for the lower MSE observed in (a). (c) Optimized signal functions $f(\theta)$ of QSP-PE (top) and RPE (bottom), with their corresponding Fourier mode decompositions shown at the lower right. The QSP-PE signal achieves near-optimal sensitivity efficiency $\kappa = 0.915$ and is considerably more sensitive than the single-mode trigonometric signals used by RPE, yielding a signal function with richer structure and a more accurate estimation.} 
\label{fig:main-figure}
\end{figure}

We further evaluate the end-to-end estimation performance and compare it with RPE. As shown in \cref{fig:main-figure}(a), the proposed QSP-based method achieves a consistently lower mean squared error (MSE), with a reduction by approximately a factor of two across different query depths. This factor-of-two gain is consistent with the sensitivity scaling: RPE has $\kappa_{\mathrm{RPE}}=1/\sqrt{2}$, while the optimized QSP signals approach $\kappa\approx 1$, and since the estimation variance scales as $\sim 1/\kappa^2$, the expected reduction is of order two.

The distribution of the estimator is shown in \cref{fig:main-figure}(b). Compared to RPE, the QSP-based estimator exhibits a more concentrated distribution around the true parameter, indicating reduced variance and improved statistical efficiency. This confirms that the gain in sensitivity directly translates into improved estimation accuracy at the level of the output distribution.

Finally, \cref{fig:main-figure}(c) illustrates the structure of the underlying signal functions. In contrast to standard approaches that rely on fixed single-mode trigonometric signals and require complementary sine and cosine transformations, QSP enables more versatile signal transformations with richer modal structure. This flexibility allows the signal to concentrate sensitivity within the relevant region of interest. Moreover, the proposed method relies on a single optimized signal rather than combining complementary signals, thereby avoiding redundancy and improving overall efficiency.

\section{Generalization to High-Dimensional Systems}

Although our analysis is carried out in the $\mathrm{SU}(2)$ setting, this reduction is a feature of QSP rather than a limitation. The qubitization structure establishes a direct correspondence between high-dimensional quantum eigenvalue transformations and an effective two-dimensional invariant subspace, enabling a simplified characterization of signal transformations. The resulting construction lifts naturally to the high-dimensional setting (see circuit in \cref{fig:main}(a)). The corresponding signal admits an explicit representation of the form $f(\mc{H}) = \braket{ \psi | g^2(\mc{H}/2) | \psi }$, where $\mc{H}$ is the underlying Hamiltonian queried through time evolution operator, which can be implemented via Trotterization \cite{childs_theory_2021}. This shows that the proposed framework is not restricted to low-dimensional problems, but applies broadly to eigenvalue estimation in general quantum systems.

Consequently, our tractable analysis and the resulting algorithms extend directly to high-dimensional settings. This allows our approach to be applied both to small-scale tasks such as gate learning, and to large-scale problems such as Hamiltonian eigenvalue estimation.

\section{Conclusions and outlook}

In this work, we introduce a new approach to phase estimation based on programmable signal design. By leveraging the flexibility of quantum signal processing, we design measurement signals with enhanced sensitivity and develop a co-designed framework that integrates quantum circuit construction with classical inference. This leads to improved sensitivity efficiency and reduced resource cost compared to standard phase-estimation methods.

Our analysis establishes a direct connection between signal design and estimation performance through the sensitivity efficiency $\kappa$, and demonstrates how optimized signal transformations translate into practical gains in precision. Numerical results confirm these improvements and illustrate the underlying mechanisms.

Beyond the specific setting studied here, our results highlight the broader potential of quantum algorithmic tools for shaping quantum sensing protocols. This perspective opens new directions for designing quantum algorithms that directly optimize information extraction in estimation and metrology tasks.

\begin{acknowledgments}
The authors thank Zhiyan Ding and Xiaodi Wu for helpful discussions. S.L. was supported by the Air Force Office of Scientific Research under award number FA9550-21-1-0209.
\end{acknowledgments}

\section*{Data availability}

The data that support the findings of this article are openly available at \footnote{\url{https://github.com/dongsnaq/QPE-Phase-Estimation}}.

\bibliographystyle{plainnat}
% \bibliography{ref_dong,ref}

\newpage 
\clearpage
\appendix
\thispagestyle{empty}
\onecolumngrid

%TC:ignore

 \begin{center}
     {\Large \bf Supplementary Material for\\
     Programmable Signal Design for Quantum Phase Estimation via Quantum Signal Processing}
 \end{center}
\begin{center}
Zikang Jia,$^1$ Suying Liu,$^{2,3}$, and Yulong Dong$^{4}$\\
\smallskip
\small{\emph{$^1$\UMMath\\$^2$\UMDCS\\$^3$\QUICS\\$^4$\UMEECS}}\\
(Dated: \today)
\end{center}

\setcounter{equation}{0}
\setcounter{figure}{0}
\setcounter{table}{0}
\setcounter{page}{1}
\setcounter{section}{0}
\setcounter{algorithm}{0}
\setcounter{secnumdepth}{3}

\makeatletter
\@removefromreset{equation}{section}
\@removefromreset{figure}{section}
\@removefromreset{table}{section}

\renewcommand{\theequation}{S\arabic{equation}}
\renewcommand{\thefigure}{S\arabic{figure}}
\renewcommand{\thealgorithm}{S\arabic{algorithm}}
\renewcommand{\thesection}{S\arabic{section}}
\renewcommand{\bibnumfmt}[1]{[S#1]}
\makeatother

\section{Signal Model and Preliminaries}

This section introduces the signal model and setup used throughout the appendix. We first review the signal representation arising from the Hadamard-test-based phase estimation, which serves as the standard reference point. We then describe how quantum signal processing (QSP) modifies this signal through programmable transformation functions, and finally formalize the resulting signal class and optimization problem for QSP-based phase estimation.

\subsection{Hadamard-test-based quantum phase estimation and its signal representation}

Standard methods for quantum phase estimation (QPE), as illustrated by the circuit in \cref{fig:main}, are based on the Hadamard test. In this construction, repeated queries to the Hamiltonian time-evolution operator are sandwiched between Hadamard gates acting on an ancilla qubit, up to an additional phase gate for estimating the imaginary component. The corresponding measurement probability in the general $n$-qubit setting is given in \cref{sec:high-d-extension}.

To expose the underlying signal structure, consider the case where the input state is an eigenstate $\ket{\psi}$ of the Hamiltonian, satisfying $\mathcal{H}\ket{\psi}=\lambda\ket{\psi}$. One can then define a qubitized subspace $\mathcal{S}_\lambda := \mathrm{span}\{\ket{0}\ket{\psi},\,\ket{1}\ket{\psi}\}$. Restricted to this subspace, the action of the controlled unitary reduces to an effective $Z$ rotation. Consequently, an effective $\mathrm{SU}(2)$ model is obtained as
\begin{equation}
    \mc{W}_{\Re}(\theta) = \mathrm{H} e^{\I d \theta Z} \mathrm{H},\quad \mc{W}_{\Im}(\theta) = \mathrm{H} \mathrm{S}^\dagger e^{\I d \theta Z} \mathrm{H}.
\end{equation}
A direct calculation gives the corresponding measurement probabilities
\begin{equation}
    f_{\Re}(\theta)=\left|\braket{0|\mathcal{W}_{\Re}(\theta)|0}\right|^2=\frac{1}{2}\bigl(1+\cos(2d\theta)\bigr), \quad f_{\Im}(\theta)=\left|\braket{0|\mathcal{W}_{\Im}(\theta)|0}\right|^2=\frac{1}{2}\bigl(1+\sin(2d\theta)\bigr).
\end{equation}

Since the estimation variance depends quadratically on the local derivative, the effective sensitivity of this paired signal family is naturally characterized by the root-mean-square average
\begin{equation}
    \kappa_{\mathrm{eff}}:=\frac{1}{d}\sqrt{\frac{\bigl(f'_{\Re}(\theta)\bigr)^2+\bigl(f'_{\Im}(\theta)\bigr)^2}{2}}=\frac{1}{\sqrt{2}}.
\end{equation}

The complementary cosine and sine signals avoid blind spots in sensitivity: although each individual signal can have vanishing derivative in certain parameter regimes, the paired signal family maintains a uniform effective sensitivity. This makes Hadamard-test-based signals a standard choice in many QPE algorithms. Methods such as Robust Phase Estimation (RPE) build on this fixed signal family and progressively refine the estimate to achieve Heisenberg-limited scaling.

At the same time, the effective sensitivity efficiency remains bounded by the universal constant $1/\sqrt{2}$. More importantly, the signal family itself is fixed throughout the iterative estimation process, so prior information from earlier stages is not used to adapt the signal. As a result, a constant gap from the optimal value is preserved due to the use of a fixed complementary signal pair. This naturally raises the question of whether one can replace this redundant fixed construction by a single adaptive signal tailored to the current uncertainty region. The rest of this Supplementary Material addresses this question.

\subsection{Transformation functions of quantum signal processing based phase estimation (QSP-PE)}

QSP is a powerful quantum algorithmic primitive for performing universal function transformations of an input Hamiltonian \cite{low_optimal_2017,gilyen_quantum_2019,martyn_grand_2021}. The algorithm interleaves oracles to the Hamiltonian with a sequence of phase-modulation gates whose angles encode the desired transformation. A key feature of QSP is that these angles depend only on the target transformation function and not on the specific Hamiltonian oracle. As a result, even for an arbitrarily high-dimensional Hamiltonian matrix, the essential QSP structure can be reduced to an approximation theorem in $\mathrm{SU}(2)$, where interleaved $Z$- and $X$-rotations carry the information of Hamiltonian eigenvalues and the transformation function, respectively in some invariant subspace \cite{gilyen_quantum_2019,wang_energy_2022}.

In the rest of this section, we first work with the $\mathrm{SU}(2)$ model of QSP, known as qubitized form, and show how to design optimal signals for phase estimation. In a later section, we discuss how to lift these qubitized constructions to arbitrary-dimensional Hamiltonians.

\begin{defn}[$\mathrm{SU}(2)$ QSP model ($Z$ convention)]\label{def:qsp_model}
Let $\Phi := (\varphi_0,\varphi_1,\dots,\varphi_d) \in \RR^{d+1}$ be a sequence of phase factors. For any real $\theta \in [0,\pi]$, the QSP sequence is defined as
\begin{equation}
    \mc{Q}(\theta,\Phi)
    = e^{\I \varphi_0 X} \prod_{j = 1}^d \big( e^{\I \theta Z} e^{\I \varphi_j X} \big).
\end{equation}
We say that a set of phase factors $\Phi$ is \emph{Z-symmetric} if it satisfies $\varphi_d = \varphi_0 + \pi / 2$ and $\varphi_j = \varphi_{d-j}$ for all $1 \le j \le d - 1$.
\end{defn}

The heart of QSP is the following representation theorem, which shows that the product $\mc{Q}(\theta,\Phi)$ generates a functional path in $\theta$ with universal expressive power.

\begin{thm}[Representation of QSP]\label{thm:qsp_rep}
Suppose $\Phi \in \RR^{d + 1}$ is a Z-symmetric set of phase factors. Then there exist real polynomials $u,v,w$ such that
\begin{equation}
    \mc{Q}(\theta,\Phi)
    = \begin{pmatrix}
        u(\cos\theta) &
        - \sin\theta\, w(\cos\theta) + \I v(\cos\theta) \\
        \sin\theta\, w(\cos\theta) + \I v(\cos\theta) &
        u(\cos\theta)
      \end{pmatrix}.
\end{equation}
These polynomials satisfy:
\begin{enumerate}
    \item[] (A1) Degree: $u$ and $v$ have degree at most $d$, and $w$ has degree at most $d-1$.
    \item[] (A2) Parity: $u$ and $v$ have parity $d \bmod 2$, and $w$ has parity $(d-1) \bmod 2$.
    \item[] (A3) Normalization: for all $\theta \in [0,\pi]$,
    \begin{equation}
        u^2(\cos\theta) + v^2(\cos\theta) + \sin^2\theta\, w^2(\cos\theta) = 1.
    \end{equation}
\end{enumerate}
Conversely, given an integer $d$ and a real polynomial $u$ that satisfies:
\begin{enumerate}
    \item[](B1) Degree: $\deg(u) = d$,
    \item[](B2) Parity: $u$ has parity $d \bmod 2$,
    \item[](B3) Boundedness: $\abs{u(\cos\theta)} \le 1$ for all $\theta \in [0,\pi]$,
\end{enumerate}
there exists a Z-symmetric set of phase factors $\Phi$ such that
\begin{equation}\label{eq:0to0_signal}
    \braket{0 | \mc{Q}(\theta,\Phi) | 0} = u(\cos\theta)
\end{equation}
for all $\theta \in [0,\pi]$.
\end{thm}
We remark that there are two common conventions in the literature. The $Z$-convention described above is typically used in algorithms based on Hamiltonian time evolution, such as QETU and Generalized QSP (GQSP). Alternatively, the standard QSP form is often written in the $X$-convention, where the $\theta$-dependent rotation acts about the $X$-axis rather than the $Z$-axis. These two conventions are equivalent, since $X$ and $Z$ rotations are exchanged under conjugation by the Hadamard gate.

\begin{proof}
Based on the discussion above, it suffices to construct a mapping between the $Z$-convention QSP and that in the standard $X$ convention. Then, the main claims in the theorem follow immediately from the results in the standard literature \cite{gilyen_quantum_2019,wang_energy_2022}.

In the $X$ convention, the set of phase factors $\Psi = (\psi_0, \psi_1, \cdots, \psi_d)$ is \emph{symmetric} in the sense that $\psi_j = \psi_{d - j}$ for any $j = 0, 1, \cdots, d$. According to the representation of QSP in the $X$ convention, for any set of symmetric phase factors $\Psi \in \RR^{d+1}$, there exist three real polynomials $u, v, w$ satisfying conditions (A1-3) such that the following result holds
\begin{equation}
    \mathcal{X}(\theta, \Psi) := e^{\I \psi_0 Z} \prod_{j = 1}^d \big( e^{\I \theta X} e^{\I \psi_j Z} \big) = \begin{pmatrix}
    u(\cos\theta) + \I v(\cos\theta) &
    \I \sin\theta\, w(\cos\theta) \\
    \I \sin\theta\, w(\cos\theta) &
    u(\cos\theta) - \I v(\cos\theta)
  \end{pmatrix}.
\end{equation}
Now, we consider the following adjusted set of phase factors
\begin{equation}\label{eqn:Z-X-phase-factor-relation}
    \Phi = (\varphi_0, \varphi_1, \cdots, \varphi_{d-1}, \varphi_d) = (\psi_0 - \pi/4, \psi_1, \cdots, \psi_{d-1}, \psi_d + \pi/4).
\end{equation}
Then, we have
\begin{equation}
    \begin{split}
        \mc{Q}(\theta, \Phi) &= \mathrm{H} \mc{X}(\theta, \Phi) \mathrm{H} = \mathrm{H} e^{- \I \pi / 4 Z} \mc{X}(\theta, \Psi) e^{\I \pi/4 Z} \mathrm{H}\\
        &= (-1)^{d/2} \begin{pmatrix}
    u(\cos\theta) &
    - \sin\theta\, w(\cos\theta) + \I v(\cos\theta) \\
    \sin\theta\, w(\cos\theta) + \I v(\cos\theta) &
    u(\cos\theta)
  \end{pmatrix}.
    \end{split}
\end{equation}
Note that the prefactor $(-1)^{d/2}$ is a global phase, which does not affect any measurement result. We can also absorb it into the definition of the polynomials $u, v, w$ for consistency.

This completes the mapping between the two conventions and hence completes the proof. Furthermore, \cref{eqn:Z-X-phase-factor-relation} also provides a practical way to derive phase factors. Given a target polynomial $g(\theta) = u(\cos\theta)$ constructed by solving the optimal design problem, we can derive the $X$-convention phase factors $\Psi$ by using \textsf{QSPPACK} \cite{dong_efficient_2021,QSPPACK}. Then, we can use \cref{eqn:Z-X-phase-factor-relation} and adjust the two edge phase factors to derive the set of phase factors used in the circuit construction.
\end{proof}

\subsection{Signal class and problem formulation of QSP-PE}

As stated in \cref{thm:qsp_rep}, the corresponding signal class consists of cosine transforms with definite parity, given by bounded trigonometric polynomials. This can be equivalently written as follows
\begin{equation}\label{eq:signal-parametrization}
    \mc{G}_d := \left\{ g(\theta) = \sum_{k = 0}^d a_k \cos(k\theta) : a_k = 0 \ \forall\, k \not\equiv d \pmod{2}, \ \text{and } \abs{g(\theta)} \le 1 \ \forall\, \theta \right\}.
\end{equation}

The measurement probability is given by the square of the transformed value. Hence, we define the admissible signal class for our estimation task as follows.

\begin{defn}[Admissible signal class]
The admissible signal class is defined as
\begin{equation}
\mathcal{F}_d 
:= \left\{ f : f(\theta) = g^2(\theta),\ g\in\mathcal{G}_d\right\}.
\end{equation}
Note that for $f\in\mathcal{F}_d$, we have $f(\theta)\in[0,1]$ for all $\theta$.
\end{defn}

This set defines the design space of admissible estimation protocols. The performance of an estimator is typically characterized in terms of its \emph{sensitivity}, namely the derivative of the signal function in a neighborhood of the ground-truth parameter. This intuition can be justified as follows.

In the estimation process, due to the uncertainty arising from finite measurement shots, the estimate of the parameter is also subject to uncertainty. By locally linearizing the measurement signal, we have the following error propagation
\begin{equation}
    s = f(\theta) \Rightarrow \Delta \theta \approx \frac{\Delta s}{f^\prime(\theta)}.
\end{equation}
In the absence of the exact parameter value, the only information we have with high confidence is the confidence interval $\theta \in \mc{I}$. Consequently, it suffices to maximize the sensitivity $\abs{f^\prime(\theta)}$ in order to minimize the estimation variance
\begin{equation}
    \mathrm{Var}(\hat{\theta}) \propto \frac{1}{\min_{\theta \in \mc{I}} (f^\prime(\theta))^2}.
\end{equation}

This leads to the following optimal design principle, which can be formulated as a max-min optimization problem over the admissible signal class:
\begin{equation}\label{eqn:max-min-optimal-condition}
    f^\star_{d, \mc{I}}
    = \arg\max_{f \in \mc{F}_d}
      \min_{\theta \in \mc{I}} \abs{f^\prime(\theta)}.
\end{equation}
This max-min optimization problem suggests that the optimal signal maximizes the worst-case sensitivity over the entire confidence interval. We will analyze the structure of this problem and develop efficient solution methods in the next section.

Before concluding this section, we present an upper bound on the sensitivity. Here, $f(\theta) = g^2(\theta)$, where $g$ is a trigonometric polynomial of degree at most $d$, and its supremum norm is bounded by one as a consequence of the unitarity of the quantum circuit. Then, by applying the Bernstein-Szeg\"{o} inequality \cite{leonteva_bernsteinszego_2022}, the following holds
\begin{equation}
    \abs{f'(\theta)} = 2 \abs{g(\theta)} \abs{g'(\theta)} \le 2d \abs{g(\theta)} \sqrt{1 - g^2(\theta)} \le 2d \cdot \frac{1}{2} = d.
\end{equation}
Therefore, we normalize the sensitivity by the upper bound $d$ when defining the sensitivity efficiency
\begin{equation}
    \kappa(f, \mc{I}) := \frac{\min_{\theta \in \mc{I}} \abs{f^\prime(\theta)}}{d}.
\end{equation}

\section{Optimal Signal Design}

This appendix provides additional details on the optimal signal design problem introduced in the main text. We first analyze the structure of the max-min formulation and characterize a necessary relation between the degree parameter and the radius of the confidence interval. We then present a relaxed optimization reformulation that enables efficient numerical solution. Finally, we report numerical results illustrating the behavior of the optimized signals and validating the theoretical results.

\subsection{Structure of the optimal design problem}

The optimal design principle in \cref{eqn:max-min-optimal-condition} is formulated as a max-min optimization over the admissible signal class $\mathcal{F}_d$. Since any admissible signal is the square of a transformation function $g \in \mc{G}_d$, the derivative of $f$ admits a special structure. We restate this optimality principle in the following definition.

\begin{defn}[Optimal design principle]\label{def:max-min}
    Given an integer $d$ and a confidence interval $\mc{I} = [\theta_0 - \rads, \theta_0 + \rads]$ of radius $\rads$ centered at a prior $\theta_0$, the optimal signal function solves the following max–min design problem:
    \begin{equation}\label{eq:max-min-square-sig}
        \squares(d,\mc{I})
        := \max_{f\in\mathcal{F}_d} \min_{\theta\in \mc{I}} \abs{f'(\theta)}
        = \max_{g\in\mathcal{G}_d} \min_{\theta\in \mc{I}} 2 \abs{g(\theta) g^\prime(\theta)}.
    \end{equation}
\end{defn}

We observe that the sensitivity depends not only on the derivative $g'(\theta)$ of the trigonometric polynomial, but also on the value of $g(\theta)$ itself. In particular, even if $\abs{g'(\theta)}$ stays large, the factor $g(\theta)$ can completely suppress the signal derivative whenever $g$ crosses zero. This suggests that we should control $g(\theta)$ and $g'(\theta)$ simultaneously in order to ensure that the signal sensitivity is uniformly bounded away from trivial zero on the entire confidence interval $\mc{I}$.

Intuitively, a good design should make both the amplitude and the derivative of the transformation function $g$ large. In other words, on $\mc{I}$ the function $g$ is both large in magnitude and relatively steep. However, a bounded transformation function that is simultaneously large and steep cannot stay on a very long interval. If $\abs{g(\theta)}$ is constrained to lie in $[1 - \eta, 1]$ for all $\theta \in \mc{I}$, then the total variation of $g$ on $\mc{I}$ is at most $\eta$. On the other hand, if $\abs{g'(\theta)}$ is of order $d$ and the interval has radius $\rads$, a mean-value argument shows that $g$ must change by an amount comparable to $2 d \rads$ across $\mc{I}$. Once $\rads$ becomes much larger than $\eta/(2 d)$, these two requirements become incompatible: either the amplitude $\abs{g}$ or the slope $\abs{g^\prime}$ must drop.

This heuristic explains why the confidence interval cannot be arbitrarily wide and should have a natural scale $\rads = \Or(1/d)$. The next lemma formalizes this necessary scaling condition.

\begin{lem}[Necessary scale for optimal designs]\label{lem:interval-scale-derivative-band}
Let $g \in \mc{G}_d$ be a transformation function and $\mc{I}$ be a confidence interval with radius $\rads$. Let $\alpha, \gamma \in (0, 1)$ be two fixed parameters. Suppose the transformation function satisfies both
\begin{equation}\label{eq:amp-band}
    1 \ge \abs{g(\theta)} \ge \alpha \text{ and } \abs{g'(\theta)} \ge \gamma d
    \quad \text{for all } \theta \in \mc{I}.
\end{equation}
Then the radius $\rads$ must satisfy
\begin{equation}\label{eq:K-necessary}
    \rads \le \frac{1- \alpha}{2 \gamma d}.
\end{equation}
Equivalently, if $\rads > \frac{1 - \alpha}{2 \gamma d}$, the amplitude and derivative constraints in \eqref{eq:amp-band} cannot hold simultaneously.
\end{lem}

\begin{proof}
Because $\abs{g(\theta)} \ge \alpha > 0$ on $\mc{I}$, the function $g$ does not change sign on $\mc{I}$. Without loss of generality, we assume $g > 0$ on $\mc{I}$. Otherwise, we can replace $g$ by $-g$ without affecting the conclusion. 

The derivative constraint implies that $g$ is monotonic. Hence, its maximum and minimum on $\mc{I}$ are attained at the endpoints. 
From the amplitude constraint, we obtain
\begin{equation}\label{eq:amp-diff-upper}
    \abs{g(\theta_0 - \rads) - g(\theta_0 + \rads)} = \max_{\theta \in \mc{I}} g(\theta) - \min_{\theta \in \mc{I}} g(\theta) \le 1 - \alpha.
\end{equation}
On the other hand, using Mean Value Theorem, there exists $\xi \in (\theta_0 - \rads,\theta_0 + \rads)$ such that
\begin{equation}
    \abs{g(\theta_0 - \rads) - g(\theta_0 + \rads)} = 2 \rads \abs{g^\prime(\xi)} \ge 2 \gamma d \rads.
\end{equation}
Combining these inequalities gives
\begin{equation}
    2 \gamma d \rads \le 1 - \alpha.
\end{equation}
This completes the proof.
\end{proof}

In the rest of the paper, we set $\alpha = 1 /2$. Note that $\gamma < 1$. It suffices to set the radius of the confidence interval to
\begin{equation}\label{eq:CI-d-relation}
    \rads = \frac{1}{4 d},
\end{equation}
so the above requirement is satisfied.

\subsection{Relaxed optimization reformulation}

A brute-force way to tackle the optimal design principle in \cref{def:max-min} is to parametrize the transformation function $g$ by its Fourier coefficients. One can then enforce a uniformly lower bounded sensitivity by imposing inequality constraints on a dense set of sampled points in $\mc{I}$. Since the signal function $f$ is given by the square of $g$, it depends quadratically on these coefficients. This leads to a nonconvex Quadratically Constrained Quadratic Program (QCQP), which is numerically difficult to solve and does not scale well with $d$ \cite{boyd2004convex}.

Instead, we derive the solution from a surrogate of the original optimization problem. As analyzed in the last subsection, a sufficient condition for high sensitivity is that the amplitude of the transformation function is bounded away from zero and its derivative is uniformly large. This motivates a reformulation of the original optimal design principle in \cref{def:max-min}: we separately control the uniform amplitude and the uniform derivative of a \emph{single} transformation function $g$ on the confidence interval by introducing auxiliary variables.

\begin{defn}[Relaxed optimal design principle]\label{def:relaxed_optimization}
Given an integer $d$, a parameter $\alpha \in (0,1)$, and a confidence interval $\mc{I}$, let $\wt{g}_{d,\mc{I},\alpha} \in \mc{G}_d$ solve
\begin{equation}
    \begin{split}
        \mathrm{maximize}\quad     & \min_{\theta \in \mc{I}} \abs{g'(\theta)} \\
        \mathrm{subject to}\quad   & g \in \mc{G}_d,\ \min_{\theta \in \mc{I}} \abs{g(\theta)} \ge \alpha.
    \end{split}
\end{equation}
\end{defn}

This provides a relaxation of the original optimal design principle, in the sense that it yields explicit feasible instances and hence a lower bound on the overall optimum:
\begin{equation}
    \squares(d,\mc{I})
    \ge \sup_{0 < \alpha < 1} 2 \alpha B(\wt{g}_{d,\mc{I},\alpha})
\end{equation}
where
\begin{equation}
    B(\wt{g}_{d,\mc{I},\alpha})
    := \min_{\theta \in \mc{I}} \abs{\wt{g}_{d,\mc{I},\alpha}'(\theta)}.
\end{equation}

Practically, we can introduce an explicit auxiliary variable $\beta$ to represent the uniform derivative lower bound $B(g) \ge \beta$. The optimization problem can then be written in a slack form by enforcing the constraints on $M$ sampled points in $\mc{I}$:
\begin{equation}\label{eqn:relaxed_optimization_slack_form}
    \begin{split}
        \text{maximize}\quad   & \beta  \\
        \text{subject to}\quad 
        & g \in \mc{G}_d, \\
        & \abs{g(\theta_i)} \ge \alpha,\quad \abs{g'(\theta_i)} \ge \beta,
          \quad i = 1, \dots, M,\ \theta_i \in \mc{I}.
    \end{split}
\end{equation}

The advantage of this relaxed optimization problem is now clear. By introducing the two auxiliary variables $\alpha$ and $\beta$, the original quadratic constraints in the coefficients are replaced by linear constraints on the function $g$. As a result, the challenging QCQP is reduced to a simpler problem. For each fixed value of $\alpha$, we optimize over $\beta$ and the coefficients of $g$ under linear amplitude and derivative constraints, effectively decoupling the control of the amplitude and the derivative.

Following the parametrization in Eq.~\eqref{eq:signal-parametrization}, the optimal coefficient vector $\mathbf{a}^{\star}$ is determined via an outer-loop grid search over the auxiliary variable $\alpha$. For each $\alpha$, the remaining parameters are resolved by repeating over all four sign patterns $(\iota_g, \iota_{g'}) \in \left\{+1,-1\right\}^2$, which fully linearizes the absolute value conditions $|g(\theta_i)| \geq \alpha$ and $|g'(\theta_i)| \geq \beta$ and bypasses non-smoothness. The optimal signal is finally selected as the best solution over all sign patterns and $\alpha$ candidates.

\begin{algorithm}[H]
\caption{Solving Optimally Designed Signal}
\label{alg:sqaured-sig-optimization}
\begin{algorithmic}

\Require $d$, prior center $\theta_0$, grid sizes $N_{\mathrm{amp}}, N_{\mathrm{prior}}, N_\alpha$.
\Ensure Optimal coefficients $\mathbf{a}^\star$, parameters $(\alpha^\star, \beta^\star)$.

\State Set interval radius $\rads = 1/(4d)$.
\State Construct grids $\Theta_{\text{amp}}$, $\Theta_{\text{prior}}$, and $\mathcal{A}$ as equidistant points on $[0,\pi]$, $[\theta_0 - r,\, \theta_0 + r]$, and $[0,1]$, with $N_{\text{amp}}$, $N_{\text{prior}}$, and $N_{\alpha}$ points, respectively.
\State Set $\texttt{best\_obj} \leftarrow -\infty$.
\For{each $\alpha \in \mathcal{A}$}
    \For{each sign pattern $(\iota_g, \iota_{g'}) \in \left\{+1,-1\right\}^2$}
        \State Solve the LP: maximize $2\alpha\beta$ over $\{a_k\}$ and $\beta \in [0,d]$, subject to:
        \State \quad (i) $|g(\mathbf{a}, \theta_j)| \leq 1$ for all $\theta_j \in \Theta_{\text{amp}}$
        \State \quad (ii) $\iota_g \cdot g(\mathbf{a}, \theta_i) \geq \alpha$ for all $\theta_i \in \Theta_{\text{prior}}$
        \State \quad (iii) $\iota_{g'} \cdot g'(\mathbf{a}, \theta_i) \geq \beta$ for all $\theta_i \in \Theta_{\text{prior}}$
        \If{$2\alpha\beta > \texttt{best\_obj}$}
            \State Update $\texttt{best\_obj}$ and store $(\mathbf{a}^\star, \alpha^\star, \beta^\star)$.
        \EndIf
    \EndFor
\EndFor
\State \Return $(\mathbf{a}^\star, \alpha^\star, \beta^\star)$

\end{algorithmic}
\end{algorithm}

\subsection{Numerical results of optimal signal design}
In this section, we numerically evaluate the performance of optimal signals constructed via the design principle in \cref{alg:sqaured-sig-optimization}. We demonstrate that the optimal signal is attained over a bounded confidence interval, as shown in \cref{fig:interval-shrinking}, within which the sensitivity efficiency $\kappa$ improves monotonically as the interval shrinks. This improvement is also uniform across target phases $\theta$ and consistently exceeds the RPE baseline, as shown in \cref{fig:kappa-theta-zeta-L-d-regression}. To account for this, we analyze the optimized signal in \cref{fig:f-shape-mode} and find that the QSP-PE signal is richer in Fourier modes than the simple cosine and sine signals of RPE, making it more efficient for phase estimation.

We first investigate the confidence interval radius within which the optimal signal can be reliably constructed. We solve for the optimal signal $f^{\star}_{d,\mathcal{I}} := (g^{\star}_{d,\mathcal{I}})^2$, where each experiment fixes a query depth $d$ and a confidence interval $\mathcal{I} := [\theta - r, \theta + r]$ centered at $\theta$ with radius $r$. The sensitivity efficiency is defined as $\kappa := L/d$, where $L := \min_{\theta \in \mathcal{I}} |f'(\theta)|$ denotes the minimum derivative magnitude over $\mathcal{I}$.

As suggested by Lemma \ref{lem:interval-scale-derivative-band}, obtaining a large $L$ for a degree-$d$ signal requires both $|g|$ and $|g'|$ to be sufficiently large throughout $\mc{I}$, imposing the constraint $r \leq 1/(4d)$. To investigate the behavior near this theoretical threshold, we evaluate the optimal signal across a range of interval radii with $d = 10$ and $\theta_0 = \sqrt{\pi}/2$, as illustrated in \cref{fig:interval-shrinking}.

\begin{figure}[htbp]
\centering
\includegraphics[width=0.96\textwidth]{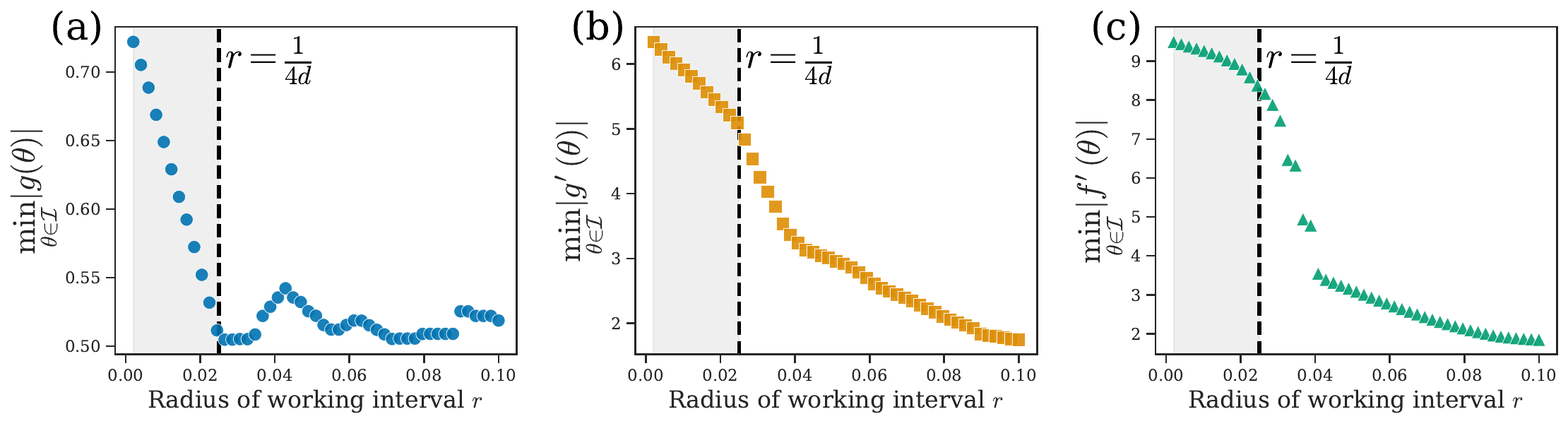}
\caption{Numerical analysis of optimized signal performance as a function of the working interval radius $r$: (a) The maximum lower bounds on $|g(\theta)|$, (b) $|g'(\theta)|$, and (c) $|f'(\theta)|$ all occur within a bounded (shaded) region of $r$.}
\label{fig:interval-shrinking}
\end{figure}

A phase transition is evident at the critical radius $r = 1/4d$. As shown in \cref{fig:interval-shrinking}(c), the magnitude $L$ exhibits a clear separation across this threshold. Within the critical radius (shaded region), $L$ is consistently larger than outside it (white region). This confirms that the confidence interval must satisfy $r \leq 1/4d$ for effective optimal signal design.

A further observation is that $L$ grows monotonically as the radius shrinks. This aligns with the intuition that optimizing over a smaller interval relaxes the max-min problem, yielding a larger minimum derivative. Motivated by this, we introduce a \textit{shrinkage factor} $\zeta$, which refines the interval radius to $r = 1/(\zeta \cdot 4d)$. A larger $\zeta$ therefore corresponds to a smaller interval and, consequently, a higher sensitivity efficiency $\kappa$.

\begin{figure}[htbp]
\centering
\includegraphics[width=0.96\textwidth]{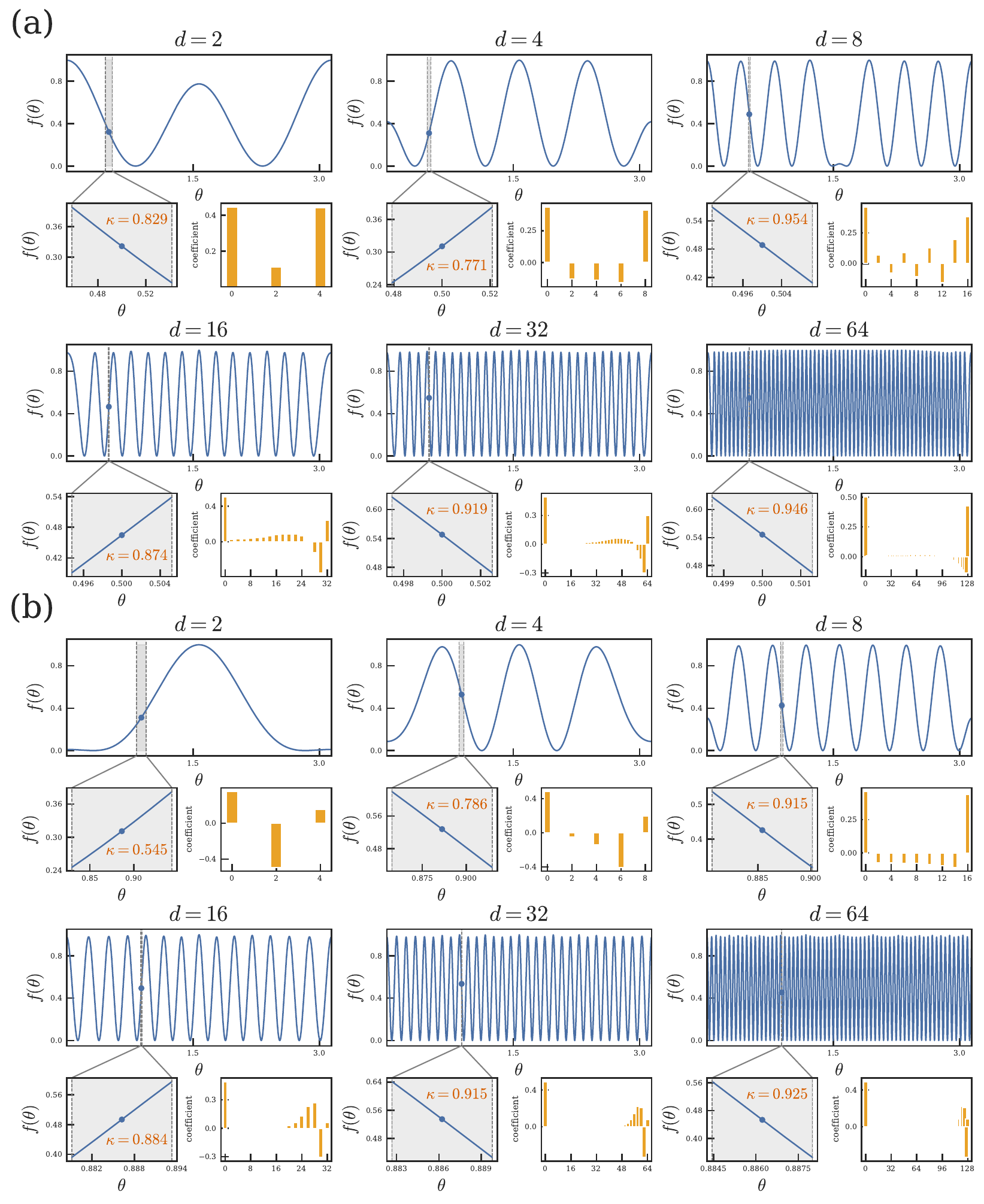}
\caption{Optimized signal $f(\theta) = g^2(\theta)$ in QSP-PE and its Fourier cosine modes at query depths $d = 2, 4, 8, 16, 32, 64$. Each panel shows (top) the signal over $\theta \in [0, \pi]$ with the prior region shaded, (bottom-left) a zoom-in near the target phase $\theta^\star$ with the normalized sensitivity $\kappa = L/d$ annotated, and (bottom-right) the Fourier cosine coefficients of $f$. (a) All signals are optimized for target phase $\theta^\star = 0.5$ with $\zeta = 2$. The sensitivity efficiency $\kappa$ averages $0.882$ across all query depths. (b) All signals are optimized for target phase $\theta^\star = \sqrt{2}/2$ with $\zeta = 2.2$. The sensitivity efficiency $\kappa$ averages $0.828$ across all query depths. Overall, the $\kappa$ results demonstrate that QSP-PE consistently achieves near-optimal signal sensitivity efficiency relative to the polynomial degree $d$.}
\label{fig:f-shape-mode}
\end{figure}

To understand the structural advantage of QSP-PE signals as discussed in \cref{fig:kappa-theta-zeta-L-d-regression} and \cref{fig:main-figure}, we examine the optimized signal $f(\theta)$ and its Fourier decomposition across query depths $d = 2, 4, 8, 16, 32, 64$ in \cref{fig:f-shape-mode}. Each panel shows three views: (top) the full signal over $\theta \in [0, \pi]$ with the prior region $\mathcal{I}$ shaded, (bottom-left) a zoom-in near the target phase $\theta^\star$ with $\kappa$ annotated, and (bottom-right) the Fourier cosine coefficients of $f$. The signals in \cref{fig:f-shape-mode}(a) are optimized for $\theta^\star = 0.5$ with $\zeta = 2$, while those in \cref{fig:f-shape-mode}(b) are optimized for $\theta^\star = \sqrt{2}/2$ with $\zeta = 2.2$.

As $d$ increases, two trends are immediately apparent. First, the signal develops richer oscillatory structure globally while maintaining a steep, nearly linear profile within the prior $\mathcal{I}$, which is precisely the behavior that maximizes $L$. This local steepness is visible in the zoom-in panels: across all six depths, the signal is monotone and well-sloped over $\mathcal{I}$, with most $\kappa$ values exceeding 0.85. Second, the Fourier coefficient panels reveal that $f$ activates modes across the full available bandwidth ${0, 2, 4, \ldots, 2d}$, with nontrivial weight distributed at both low and high frequencies. This stands in sharp contrast to the RPE reference signals $(1 + \cos(2d\theta))/2$ and $(1 + \sin(2d\theta))/2$, which concentrate at exactly two modes: the constant mode 0 and the highest frequency mode $2d$. By spreading spectral weight across intermediate modes, the QSP-PE signal can simultaneously control both the amplitude $|g|$ and the derivative $|g'|$ over the prior, achieving a larger minimum derivative $L$ than a two-mode signal of the same degree could. This Fourier richness is therefore not incidental but is the mechanism through which QSP-PE surpasses the RPE baseline $\kappa_{\text{RPE}} = 1/\sqrt{2} \approx 0.707$ uniformly across query depths and target phases.

\section{Iterative Refinement Algorithm and Analysis}

In general, an estimation signal of higher degree has a much steeper derivative. Thus, in principle, it enables a more sensitive estimate of the underlying parameter. However, directly using a high-degree signal comes with a structural cost. Higher-degree trigonometric polynomials are more oscillatory, so the likelihood landscape develops many additional local minima. This phenomenon is often referred to as \textit{aliasing}.

To fully exploit the statistical power of such signals, a common strategy is to identify the correct solution branch progressively starting from lower-degree signals. This multistage procedure refines the estimate as $d$ increases, so that the estimator moves closer to the true parameter and the estimation variance is reduced at each stage. In this section, we analyze this algorithm, showing how to schedule the depth-measurement pairs in each stage and how to extract parameter estimates via efficient classical post-processing.

\subsection{Iterative refinement scheme}

We consider an adaptive iterative estimator for a parameter $\theta$ under a resource schedule $\left\{(d_k,m_k)\right\}_{k\ge 0}$, where $d_k$ is the query depth and $m_k$ is the number of measurement shots at step $k$. The procedure iteratively refines a confidence interval for $\theta$ estimation using optimized signals with bounded derivatives and resource-aware confidence updates.

To describe the iterative refinement scheme, we first define the quantities that appear in each iteration $k$:
\begin{itemize}
    \item Query depth $d_k$ (the number of $U(\theta)$ applications) and the number of measurement shots $m_k$.
    \item Parameter estimate $\hat{\theta}_k$ and confidence radius $r_k$, so that the confidence interval is $\mc{I}_k := [\hat{\theta}_k - r_k,\ \hat{\theta}_k + r_k]$.
    \item Designed transformation function $g_k \in \mc{G}_{d_k}$ and a derivative lower bound for the corresponding signal $L_k \le \min_{x \in \mc{I}_k} \abs{f_k'(x)}$, which is read off from the optimization result.
    \item Success event $S_k = \left\{\abs{\hat{\theta}_k - \theta} \le r_k\right\}$, meaning that the true value lies in the current confidence interval, and one-step conditional failure probability $\delta_k \ge \mathbb{P}(S_k^c | S_{k - 1})$.
\end{itemize}

The refinement scheme starts with low-degree signals to obtain a coarse estimate.

We now analyze how to update these quantities in the next iteration. At step $k+1$, the information about $\theta$ is encoded in the confidence interval $\mc{I}_k$ obtained at the end of step $k$. Our goal is to design a steeper signal $f_{d_{k+1},\mc{I}_k}^\star$ on this working interval $\mc{I}_k$ using a circuit of depth $d_{k+1}$. 

According to the relation in \cref{eq:CI-d-relation}, when we set the amplitude lower bound of the signal to $1/2$, the working depth $d_{k+1}$ and the confidence radius $r_k$ at the current stage are linked by
\begin{equation}
    d_{k+1} = \left\lceil \frac{1}{4 r_k} \right\rceil.
\end{equation}
In practice, such an optimal signal can be obtained by solving the relaxed optimization problem in \cref{def:relaxed_optimization,eqn:relaxed_optimization_slack_form}. The corresponding quantum circuit is then implemented using phase factors generated by \textsf{QSPPACK} \cite{dong_efficient_2021,QSPPACK}.

The signal is derived from experiments by averaging binary measurement outcomes. With $m_{k+1}$ shots, we obtain
\begin{equation}
    s_{k+1}
    = \frac{1}{m_{k+1}} \sum_{j=1}^{m_{k+1}} x_j
    := f_{d_{k+1},\mc{I}_k}^\star(\theta) + z_{k+1},
    \quad x_j \stackrel{\mathrm{IID}}{\sim} \mathrm{Ber}(f_{d_{k+1},\mc{I}_k}^\star(\theta)).
\end{equation}

In the classical post-processing step, we update the estimate by solving a constrained least squares problem:
\begin{equation}\label{eq:classical-step}
    \hat{\theta}_{k+1}
    = \mathrm{arg}\min_{x \in \mc{I}_k} \abs{f_{d_{k+1},\mc{I}_k}^\star(x) - s_{k+1}}^2
    =: h_{k+1}(s_{k+1}).
\end{equation}
In the next subsection, we show that this problem can be solved very efficiently using a Newton's bisection method, thanks to the locally monotonic structure of the optimally designed signal.

The uncertainty of the updated estimate $\hat{\theta}_{k+1}$ can then be quantified, which leads to an updated confidence radius $r_{k+1}$. To proceed, we first bound the one-step conditional failure probability as follows.

\begin{lem}\label{lem:one-step-failure}
The conditional failure probability at step $(k+1)$, given that step $k$ succeeds, is bounded by
\begin{equation}
    \PP(S_{k+1}^c | S_k)
    \le 2 \exp\left( - 2 m_{k+1} L_{k+1}^2 r_{k+1}^2 \right).
\end{equation}
\end{lem}

\begin{proof}
On the event $S_k$ we have $\theta \in \mc{I}_k$. By \cref{prop:clif-ls}, on $\mc{I}_k$, the estimation map $h_{k+1}$ locally inverts the signal. That is $h_{k+1} \circ f_{d_{k+1},\mc{I}_k}^\star\big|_{\mc{I}_k} = \mathrm{Id}$. Therefore, conditioned on $S_k$, it holds that
\begin{equation}
    \theta = h_{k+1}\big(f_{d_{k+1},\mc{I}_k}^\star(\theta)\big).
\end{equation}
Using the Lipschitz property in \cref{lem:CIF_Lip}, we obtain
\begin{equation}
    \abs{\hat{\theta}_{k+1} - \theta}
    = \abs{h_{k+1}(s_{k+1}) - h_{k+1}(f_{d_{k+1},\mc{I}_k}^\star(\theta))}
    \le \frac{1}{L_{k+1}} \abs{z_{k+1}}.
\end{equation}
Hence, on $S_k$,
\begin{equation}
    \left\{\abs{\hat{\theta}_{k+1} - \theta} > r_{k+1}\right\} \cap S_k
    \subset \left\{\abs{z_{k+1}} > L_{k+1} r_{k+1}\right\} \cap S_k.
\end{equation}

The noise term $z_{k+1}$ depends only on the fresh measurements at step $k+1$, and is independent of $S_k$, which is determined by all previous steps. Therefore
\begin{equation}
    \PP(S_{k+1}^c | S_k) \le \PP(\abs{z_{k+1}} > L_{k+1} r_{k+1}).
\end{equation}
Since $z_{k+1}$ is the centered average of $m_{k+1}$ Bernoulli variables, Hoeffding's inequality yields
\begin{equation}
    \PP(\abs{z_{k+1}} > L_{k+1} r_{k+1})
    \le 2 \exp\left( - 2 m_{k+1} L_{k+1}^2 r_{k+1}^2 \right).
\end{equation}
This completes the proof.
\end{proof}

Given a target failure probability $\delta_{k+1}$, we choose the radius $r_{k+1}$ of the next confidence interval $\mc{I}_{k+1} := [\hat{\theta}_{k+1} - r_{k+1},\ \hat{\theta}_{k+1} + r_{k+1}]$ by inverting the bound above:
\begin{equation}\label{eq:CI}
    r_{k+1}
    = \frac{1}{L_{k+1}} \sqrt{\frac{1}{2 m_{k+1}} \log\frac{2}{\delta_{k+1}}}.
\end{equation}

Using the depth-radius relation in \cref{eq:CI-d-relation}, the update rule for the depth parameter is
\begin{equation}
    d_{k+1}
    = \left\lceil \frac{1}{4 r_k} \right\rceil
    = \left\lceil \frac{L_k}{4} \sqrt{\frac{2 m_k}{\log(2 / \delta_k)}} \right\rceil
    \approx d_k \left( \frac{1}{4} \kappa \sqrt{\frac{2 m_k}{\log(2 / \delta_k)}} \right).
\end{equation}
Here, $\kappa$ is the \textit{sensitivity efficiency}, which measures how much sensitivity (derivative) per gate use yields, namely,
\begin{equation}
    \kappa \approx L_k / d_k \le 1.
\end{equation}
This gives a recursive update of the schedule $(d_k,m_k)$.

To achieve theoretical optimality in the sense of Heisenberg-limited scaling, a natural choice is to let the depth grow geometrically. This can be enforced by setting the prefactor in the parentheses to a control constant $q$, which constrains the measurement shots to
\begin{equation}\label{eqn:minimal-m}
    m_k = \left\lceil \frac{8q^2}{\kappa^2} \log\frac{2}{\delta_k} \right\rceil.
\end{equation}
This choice uses very few measurement shots at each step, which is almost a constant number.

At the same time, the depth grows as a geometric sequence:
\begin{equation}\label{eq:geometric-depth}
    d_k = q d_{k-1} = q^k,\quad k = 1,2,\dots,K.
\end{equation}
By iteratively repeating the estimation procedure with increasingly steep signals, the confidence radius shrinks exponentially and the estimator becomes exponentially more accurate. The full procedure is summarized in Algorithm~\ref{alg:iterative-refinement}.

Remarkably, when the measurement budget exceeds the minimal requirement in \cref{eqn:minimal-m}, these yield a less uncertain measurement outcome and can be leveraged to further tighten the confidence interval. Recall that, at the minimum shot count, the confidence radius is set to $r_k = 1/(4d_{k+1})$ in the algorithm workflow. Specifically, let us consider a ratio $\zeta$ so that $m_{\mathrm{actual}} = \zeta^2 m_{\mathrm{minimal}}$. Then, the resulting confidence interval has a smaller radius $r_k = 1/(4\zeta d_{k+1})$. This narrower interval can be tailored to design signals with higher sensitivity efficiency $\kappa$, as shown in \cref{fig:interval-shrinking}.

\begin{algorithm}[H]
\caption{Iterative Refinement for QSP-PE\label{alg:iterative-refinement}}
\begin{algorithmic}
\Require Number of iterations $K$, measurement shots $(m_0,\cdots,m_K)$ for each iteration, shrinkage factor $\zeta$, initial depth $d_0$, depth growth common ratio $q$.
\Ensure Estimate $\hat \theta_K$, radius of confidence interval $r_K$.

\State Implement $e^{\I \theta Z}\ket{0}$. Measure $\ket{0}$ and obtain binary measurement outcomes $\left\{x_1,\ldots,x_{m_0}\right\}$.
\State Compute the signal $s_0 = 1/m_0\sum_i x_i$ and obtain the initial estimator $\hat{\theta}_0 = 1/2 \cdot \arccos(2s_0 - 1)$.
\State Initialize confidence radius $r_0$ and set $\mathcal{I}_0 = [\hat{\theta}_0 - r_0,\, \hat{\theta}_0 + r_0]$.

\For{$k = 1$ to $K$}

    \State Solve the max-min problem over $\mathcal{F}_{d_k}$ via \cref{alg:sqaured-sig-optimization} to obtain the optimized signal $f_{d_k,\mathcal{I}_{k-1}}^\star(\hat{\theta}_{k-1})$ and record the derivative lower bound $L_k$.

    \State Prepare the state $\mathcal{Q}(\theta, \Phi)\ket{0}$ where $\Phi$ is the corresponding QSP phase factors solved by \textsf{QSPPACK}. Measure $\ket{0}$ and obtain binary measurement outcomes $\left\{x_1,\ldots,x_{m_k}\right\}$.

    \State Compute the signal $s_k = 1/m_k\sum_i x_i$ and obtain the updated estimator $\hat{\theta}_k = h_k(s_k)$ via Newton's bisection method in \cref{eq:classical-step}.

    \State Update the confidence radius $r_k = \frac{1}{4\zeta\, d_{k+1}}$ and set $\mathcal{I}_k = [\hat{\theta}_k - r_k,\, \hat{\theta}_k + r_k]$.

\EndFor

\State \Return $\hat{\theta}_K, r_K$
\end{algorithmic}
\end{algorithm}

\subsection{Classical post-processing}
One important feature of our designed signal function is that it has a uniformly lower-bounded derivative on the confidence interval. Therefore, estimating the unknown parameter by inverting the measurement outcome is straightforward: we can recover $\theta$ by locating the intersection point on the signal curve. This is formally stated in the following theorem, where we explicitly derive the form of the estimation function $h_{k+1}$, which is originally defined in \cref{eq:classical-step} through an implicit least squares problem.

\begin{thm}\label{prop:clif-ls}
    Let $\mc{I}_k = [a_k, b_k]$ be the confidence interval, and let $f_{d_{k+1}, \mc{I}_k}^\star$ be the signal function obtained from the optimal design principle. Suppose $\mc{J}_k := f_{d_{k+1}, \mc{I}_k}^\star(\mc{I}_k)$, $J_k^{\max} := \max \mc{J}_k$ and $J_k^{\min} := \min \mc{J}_k$. The solution to the constrained least squares problem in \cref{eq:classical-step} is
    \begin{equation}\label{defn:estimation_function}
        h_{k+1}(s_{k+1}) = \left\{
        \begin{array}{ll}
            a_k, & \text{when } s_{k+1} < J_k^{\min}, \\[4pt]
            (f_{d_{k+1}, \mc{I}_k}^\star)^{-1}(s_{k+1}), & \text{when } s_{k+1} \in \mc{J}_k, \\[4pt]
            b_k, & \text{when } s_{k+1} > J_k^{\max}.
        \end{array}
        \right.
    \end{equation}
\end{thm}

\begin{proof}
    When $s_{k+1} \notin \mc{J}_k$, the value $f_{d_{k+1}, \mc{I}_k}^\star(x)$ never coincides with $s_{k+1}$ on $\mc{I}_k$, so the closest value is attained at one of the endpoints. The minimizer is therefore the endpoint $a_k$ or $b_k$ that yields the smaller distance.

    When $s_{k+1} \in \mc{J}_k$, the monotonicity of $f_{d_{k+1}, \mc{I}_k}^\star$ on $\mc{I}_k$ implies that the equation $f_{d_{k+1}, \mc{I}_k}^\star(x) = s_{k+1}$ has a unique solution in $\mc{I}_k$, which attains zero least squares error. Hence $h_{k+1}(s_{k+1})$ is given by the inverse $(f_{d_{k+1}, \mc{I}_k}^\star)^{-1}(s_{k+1})$.
\end{proof}

Thanks to the monotonicity of the signal function, when the sample lies in the working interval, $s_{k+1} \in \mc{J}_k$, the inversion of the signal can be solved efficiently by applying Newton's bisection method to the nonlinear equation
\begin{equation}
    f_{d_{k+1}, \mc{I}_k}^\star(\theta) - s_{k+1} = 0.
\end{equation}
Suppose the solution to this equation is $\hat\theta_{k+1} := h_{k+1}(s_{k+1})$. Then, after $n_{k+1}$ iterations of bisection, the numerical estimator satisfies
\begin{equation}
    \abs{\wt{\hat\theta}_{k+1} - \hat\theta_{k+1}} \le r_k / 2^{n_{k+1}}.
\end{equation}
We may iterate the solver sufficiently many times so that the numerical error is smaller than the size of the next confidence interval, for example at most $r_{k+1} / \eta$ for some $\eta > 1$. This leads to
\begin{equation}
    n_{k+1} = \lceil \log_2(\eta r_k / r_{k+1}) \rceil.
\end{equation}
Since each iteration requires the evaluation of a trigonometric polynomial of degree $d_{k+1}$, which costs $\Or(d_{k+1})$ classical operations, we obtain the following bound on the classical post-processing cost.

\begin{lem}\label{lem:classical-cost-per-step}
    Let $q = r_{k - 1} / r_k$ be the radius shrinkage parameter. At the $k$-th step, estimating $\hat{\theta}_{k}$ with numerical error at most $r_{k} / \eta$ for some $\eta > 1$ requires 
    \begin{equation}
        C_k = \Or(d_{k} n_{k}) = \Or(d_{k} \log(\eta r_{k-1} / r_{k})) = \Or(d_k \log(q \eta))
    \end{equation}
    classical operations.
\end{lem}

Another key feature of the estimation function $h_{k + 1}$ is its Lipschitz continuity.
\begin{lem}\label{lem:CIF_Lip}
    Let $L_{k + 1} := \min_{x \in \mc{I}_{k}} \abs{(f_{d_{k+1}, \mc{I}_k}^\star)^\prime(x)}$. Then $h_{k+1}$ is $1/L_{k+1}$-Lipschitz. Namely, for any $s_1, s_2 \in \RR$, we have
    \begin{equation}
        \abs{h_{k+1}(s_2) - h_{k+1}(s_1)} \le \frac{1}{L_{k+1}} \abs{s_2 - s_1}.
    \end{equation}
\end{lem}
\begin{proof}
    When $s_1, s_2 \in \mc{J}_k$, the claim follows directly from function inversion. We prove the case when one point lies outside $\mc{J}_k$; the other cases are similar. For example, consider $s_1 < J_k^{\min}$ and $s_2 \in \mc{J}_k$. Then
    \begin{equation}
        \abs{s_2 - s_1} \ge \abs{f_{d_{k+1},\mc{I}_{k}}^\star(h_{k+1}(s_2)) - f_{d_{k+1},\mc{I}_{k}}^\star(a_k)} \ge L_{k+1} \abs{h_{k+1}(s_2) - a_k} = L_{k+1} \abs{h_{k+1}(s_2) - h_{k+1}(s_1)}.
    \end{equation}
    Here, the two inequalities follow from $s_1 < J_k^{\min} \le s_2$ and the Mean Value Theorem.
\end{proof}

\section{Resource Cost Analysis}\label{app:sec:resource_analysis}

This section provides the detailed resource analysis underlying \cref{thm:QSP-PE-performance-guarantee}. We derive the quantum and classical resource costs of QSP-based phase estimation (QSP-PE), including the query complexity, measurement requirements, and classical post-processing cost under the iterative refinement scheme.

\subsection{Heisenberg-limited quantum resource cost}

Recall that the depth schedule of our quantum algorithm is a geometric sequence. As a result, the total quantum resource cost is dominated by the final step, which suggests the optimal Heisenberg-limited scaling. In this subsection, we formally demonstrate that this adaptive framework attains this optimal scaling.

\begin{thm}\label{thm:quantum-cost}
Given a target precision $\epsilon > 0$ and a failure probability $\delta \in (0,1)$, let $ K = \lfloor \log_q(1/(4\epsilon)) \rfloor$ be the number of iterations required to reach $\epsilon$. Let $\kappa$ be the sensitivity efficiency. Under the adaptive schedule with the geometric depth update
$d_k = q^k$ for $q \ge 2$ and the constant measurement update $m_k = \lceil 8q^2/\kappa^2 \log(2K/\delta) \rceil$, the quantum algorithm outputs the final estimator $\hat{\theta}$ satisfies:
\begin{equation}
\mathbb{P}(|\hat{\theta} - \theta| \le \epsilon) \ge 1 - \delta,
\end{equation}
where the total number of queries to $U(\theta)$ is
\begin{equation}
    T \le \frac{4q^2}{\kappa^2 \epsilon} \left( \log(2/\delta) + \log\log(1/(4 \epsilon)) - \log\log q \right) = \wt{\Or}\left(\frac{q^2}{\kappa^2 \epsilon}\log\left(\frac{1}{\delta}\right) \right).
\end{equation}
\end{thm}

\begin{proof}
We first use the union bound to provide the final success probability of our algorithm. Consider the success event $S_k = \{|\hat{\theta}_k-\theta|\le r_k\}$ at the $k$th step, where the confidence interval is given in Eq.~\eqref{eq:CI}. Lemma.~\ref{lem:one-step-failure} provides that the conditional failure probability at any single step is bounded by $\delta_k$. We decompose the global failure event $S_K^c$ by considering the chain of dependencies.
\begin{equation}
S_{k+1}^c = (S_{k+1}^c \cap S_k) \cup (S_{k+1}^c \cap S_k^c) \subset (S_{k+1}^c \cap S_k) \cup S_k^c.
\end{equation}
Applying the union bound, we have:
\begin{equation}
\mathbb{P}(S_{k+1}^c) \le \mathbb{P}(S_k^c) + \mathbb{P}(S_{k+1}^c | S_k) \le \mathbb{P}(S_k^c) + \delta_{k+1}.
\end{equation}
By recursively applying this relation across $K$ steps, the global failure probability is bounded as:
\begin{equation}
    \mathbb{P}(S_K^c) \le \sum_{j=1}^K \delta_j.
\end{equation}
Then it suffices to set single-step failure probability as $\delta_j = \delta / K$, which sums up to the target final failure probability $\delta$. Therefore, the final success probability is $\mathbb{P}(|\hat{\theta}_K - \theta| \le r_K) \ge 1 - \delta$.

Next, we evaluate the total quantum resource cost $T$. Given the choice of $m_k = 8(q^2/\kappa^2) \log(2K/\delta)$ and $d_k=q^k$, the total cost is 
\begin{equation}\label{eq:total-resource}
     T = \sum_{k = 1}^K d_k m_k = \sum_{k = 1}^K q^k \cdot \frac{8q^2}{\kappa^2} \log(2 K / \delta) =  \frac{8}{\kappa^2} \log(2 K / \delta) q^{K + 2} \left(\frac{q - q^{- (K-1)}}{q - 1}\right)
\end{equation}
Given the constant $q \ge 2$, $ A(q):=\frac{q - q^{- (K-1)}}{q-1} \in (1,2)$ can be treated as a constant. This leads to  
\begin{equation}
    T = \frac{8}{\kappa^2} q^{K + 2} \log(2 K / \delta) A(q).
\end{equation}
Considering the final precision at the $K$-th step, it satisfies
\begin{equation}
r_K = \frac{1}{\kappa d_K} \sqrt{\frac{1}{2 m_K} \log\frac{2 K}{\delta}} = \frac{1}{4q^{K+1}} \le \epsilon.
\end{equation}
It suffices to choose $ K = \lfloor \log_q(1/(4\epsilon)) \rfloor \ge \log_q(1/(4\epsilon))-1$ so that the last inequality holds.

Then, the total resource required to bound the accuracy by $\epsilon$ is at most 
\begin{equation}
    T \le \frac{4 q^2}{\kappa^2 \epsilon} \left( \log(2/\delta) + \log\log(1/(4 \epsilon)) - \log\log q \right)=\wt{\Or}\left(\frac{q^2}{\kappa^2 \epsilon}\log\left(\frac{1}{\delta}\right) \right).
\end{equation}
Furthermore, the complexity is of the tight order
\begin{equation}
T = \wt{\Theta}\left(\frac{q^2}{\kappa^2 \epsilon}\log\frac{1}{\delta}\right).
\end{equation}
\end{proof}
The above shows that the total resource $T$ scales with $\wt{\Or}(1/\epsilon)$, thus our quantum algorithm attains Heisenberg-limited scaling up to a very small log-log factor. Furthermore, $T$ exhibits a dependence on a prefactor $q^2 / \kappa^2$. A larger $q$ accelerates the refinement process by reducing the total number of iterations $K$ with more aggressive depth, while requiring more measurements per step. Meanwhile, improving the sensitivity efficiency $\kappa$ increases the effective utilization of the query and reduces the algorithmic cost. Our algorithm not only achieves asymptotic optimality but also offers practically optimized solutions considering the realistic resource constraints with flexible $(d_k, m_k)$-schedule adjustment. 

\subsection{Classical computational cost of post-processing}

Thanks to the local monotonicity of the signal function, the classical post-processing can be carried out efficiently via bisection method. As shown in \cref{lem:classical-cost-per-step}, each post-processing step requires a total number of basic floating-point operations that scales linearly with the depth parameter. Combining this with the quantum experiment schedule, we can derive the total classical computational cost.

\begin{thm}
    Let $\epsilon$ be the target estimation precision, let $\eta$ be the relative precision required in each post-processing step, and let $q \ge 2$ be the radius shrinkage parameter. Under the same assumptions as in \cref{thm:quantum-cost}, in terms of floating-point operations, the total classical computational cost of post-processing is
    \begin{equation}
        C_\mathrm{total} = \Or\left(\frac{\log(\eta q)}{\epsilon}\right).
    \end{equation}
\end{thm}

\begin{proof}
    Summing the stepwise cost from \cref{lem:classical-cost-per-step}, we obtain
    \begin{equation}
        C_\text{total} = \sum_{k = 1}^K C_k = \Or\left( \log(q \eta) \sum_{k = 1}^K d_k \right).
    \end{equation}
    In \cref{eq:total-resource}, we showed that
    \begin{equation}
        \sum_{k = 1}^K d_k = A(q) q^K \le \frac{A(q)}{4 \epsilon} = \Or\left(\frac{1}{\epsilon}\right).
    \end{equation}
    Combining the two bounds proves the result.
\end{proof}

The bisection method has logarithmic complexity for solving a nonlinear equation to a target accuracy. Meanwhile, evaluating a general trigonometric polynomial at one point costs linearly in the degree parameter $d$. Therefore, the cost bound in \cref{lem:classical-cost-per-step} is arguably tight up to constants. As a result, the total classical computational cost of post-processing, measured in floating-point operations, also achieves the tight $1/\epsilon$ scaling.

\section{Alternative Measurement Basis}\label{sec:full-range-coverage}
In the previous sections, we assumed that the signal function is obtained by measuring the transition probability from the initial state $\ket{0}$ to the measurement outcome $\ket{0}$. This leads to the signal class $f(\theta) = g^2(\theta)$ with $g \in \mc{G}_d$. Thanks to the universality of QSP, this gives a highly flexible representation of signal functions, and the sensitivity can be improved by designing an optimal signal within this class.

However, this choice of signal function also has a limitation. Since the resulting signal is represented through a cosine series, when $\theta$ is close to zero its derivative becomes uniformly small because of the local behavior of cosine near zero. As a result, the sensitivity of this signal class is limited in regimes where $\theta$ is close to a singular point such as zero. This issue can be resolved by changing the initialization or measurement basis. In this section, we show that this can be done by measuring the transition probability from $\ket{0}$ to $\ket{+}$ instead. For this protocol, the measurement probability can be derived from the QSP representation in \cref{thm:qsp_rep}:
\begin{equation}\label{eq:0to+}
\begin{split}
     h(\theta)
     &= \abs{\braket{+ | \mc{Q}(\theta,\Phi) | 0}}^2 = \abs{\frac{1}{\sqrt{2}}\left(u(\cos\theta) + \sin\theta\, w(\cos\theta) + \I v(\cos\theta)\right)}^2 \\
     &= \frac{1}{2}\left(2 \sin\theta\, w(\cos\theta) u(\cos\theta) + u^2(\cos\theta) + v^2(\cos\theta) + \sin^2\theta\, w^2(\cos\theta)\right) \\
     &= \frac{1}{2} + \sin\theta\, w(\cos\theta) u(\cos\theta).
\end{split}
\end{equation}

When the phase factors are chosen as $\Phi = 0 \in \RR^{d+1}$, we have
\begin{equation}
    h_0(\theta)
    = \frac{1}{2} + \sin\theta\, U_{d-1}(\cos\theta)\, T_d(\cos\theta)
    = \frac{1}{2} + \sin(d\theta)\cos(d\theta)
    = \frac{1}{2} + \frac{1}{2}\sin(2d\theta).
\end{equation}

Its derivative near zero is uniformly bounded below as follows. When $\mc{I} \subset (-\pi / (4d), \pi / (4d))$ and $r := \max_{\theta \in \mc{I}} \abs{\theta}$, we have
\begin{equation}
    \min_{\theta \in \mc{I}} \abs{h_0'(\theta)} \ge d \cos(2dr).
\end{equation}
Thus, when the working interval is close to zero, even this simple choice of phase factors already gives very good performance. This also suggests that one can optimize not only the circuit structure through the phase factors, but also the initialization and measurement basis when designing an optimal protocol.

It is also worth noting that, in applications involving high-dimensional Hamiltonians, this issue can be avoided by shifting the Hamiltonian, equivalently by adding a phase to the evolution operator.

\section{Extension to High Dimensions}\label{sec:high-d-extension}
\begin{figure}
    \centering
    \begin{center}
        \begin{quantikz}[row sep=0.6cm, column sep=0.45cm]
            \lstick{$\ket{0}$}
            & \gate{e^{\I\varphi_d X}}
            & \ctrl{1}
            & \gate{e^{\I\varphi_{d-1} X}}
            & \ctrl{1}
            & \push{\cdots}
            & \ctrl{1}
            & \gate{e^{\I\varphi_1 X}}
            & \ctrl{1}
            & \gate{e^{\I\varphi_0 X}}
            & \meter{} \\
            
            \lstick{$\ket{\psi}$}
            & \qw
            & \gate{O_\mc{H}}
            & \qw
            & \gate{O_\mc{H}}
            & \push{\cdots}
            & \gate{O_\mc{H}}
            & \qw
            & \gate{O_\mc{H}}
            & \qw
            & \qw
        \end{quantikz}
    \end{center}
    \caption{Quantum eigenvalue transformation circuit for implementing high-dimensional extensions of QSP-PE.}
    \label{fig:qetu}
\end{figure}
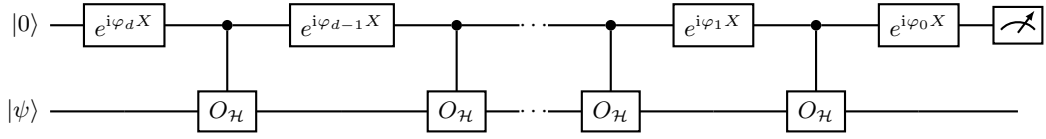

We start from the signal representation induced by a general $n$-qubit standard QPE circuit, as depicted in \cref{fig:main}. A direct calculation shows that the measurement probabilities are
\begin{equation}
    p_{\Re}(\mathcal{H}) = \frac{1}{2} + \frac{1}{2} \braket{\psi | \cos(d \mathcal{H}) | \psi}, 
    \quad
    p_{\Im}(\mathcal{H}) = \frac{1}{2} + \frac{1}{2} \braket{\psi | \sin(d \mathcal{H}) | \psi}.
\end{equation}
When the input state is an eigenstate, $\mathcal{H}\ket{\psi} = \lambda \ket{\psi}$, the dynamics reduce to a qubitized subspace $\mathcal{S}_\lambda := \mathrm{span}\{\ket{0}\ket{\psi}, \ket{1}\ket{\psi}\}$, where the evolution acts as an effective $Z$ rotation $e^{\I \theta Z}$ with $\theta = \lambda/2$. This argument recovers the result derived in $\mathrm{SU}(2)$.

We now present the general signal representation induced by a QSP-PE circuit. As discussed earlier, QSP provides an algorithmic framework that connects an approximation theory in $\mathrm{SU}(2)$ with function transformations of high-dimensional matrices. In this section, we show how to lift the function-level analyses in this appendix to a QPE procedure for large matrices.

Assume access to a controlled time evolution $O_\mc{H} := e^{-\I \mc{H}}$, where $\mc{H}$ has been shifted and rescaled by adding phase gates and tuning evolution time so that $0 \prec \mc{H} \prec \frac{\pi}{2} I$. A particular QSP circuit built from this Hamiltonian input model is shown in \cref{fig:qetu}. It is worth noting that this circuit differs from those used in quantum algorithmic applications \cite{dong_ground-state_2022,motlagh_generalized_2024} in that it employs only forward time evolution. The transformation implemented by this circuit is characterized by the following result.

\begin{thm}\label{thm:qetu_mat_rep}
Let $\Phi \in \RR^{d+1}$ be a set of phase factors, and let $(u,v,w)$ be the polynomials generated by the associated QSP sequence as in \cref{thm:qsp_rep}. Let $\ket{\psi}$ be an input quantum state. Denote by $p_{0 \to 0}$ and $p_{0 \to +}$ the probabilities that, starting from ancilla state $\ket{0}$, the measurement of the ancilla is projected onto $\ket{0}$ and $\ket{+}$, respectively, for the circuit in \cref{fig:qetu}. Then
\begin{equation}
    p_{0 \to 0}
    = \braket{\psi | u^2(\cos(\mc{H}/2)) | \psi}
    \quad \text{and} \quad
    p_{0 \to +}
    = \frac{1}{2}
      + \braket{\psi | \sin(\mc{H}/2)\, u(\cos(\mc{H}/2))\, w(\cos(\mc{H}/2)) | \psi}.
\end{equation}
In other words, the measurement statistics follow the same functional structure as in the scalar case, with the scalar variable replaced by the matrix $\mc{H}/2$.
\end{thm}

\begin{proof}
First note that applying the unitary $O_\mc{H}^{-d/2}$ on the system register at the end of the circuit does not change the measurement probability of the ancilla. It is therefore equivalent to analyze the circuit on the left-hand side of \cref{fig:qetu_proof}, whose matrix representation we denote by $\wt{\mc{U}} (\mc{H},\Phi)$. This is related to the matrix representation $\mc{U} (\mc{H},\Phi)$ of the original circuit by
\begin{equation}\label{eqn:qetu_circ_rep_relation}
    \mc{U} (\mc{H},\Phi)
    = (I \otimes O_\mc{H}^{d/2})\, \wt{\mc{U}} (\mc{H},\Phi).
\end{equation}

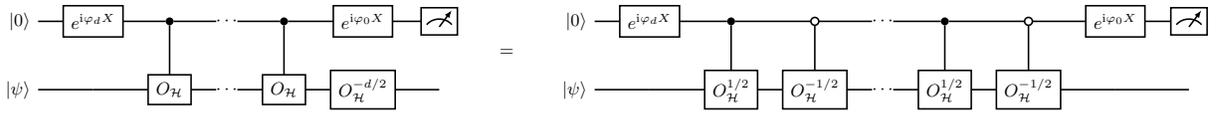
\begin{figure}[H]
    \centering
    \resizebox{0.9\textwidth}{!}{%
    \begin{quantikz}[row sep=0.6cm, column sep=0.45cm]
    \lstick{$\ket{0}$}
    & \gate{e^{\I\varphi_d X}}
    & \ctrl{1}
    & \push{\cdots}
    & \ctrl{1}
    & \gate{e^{\I\varphi_0 X}}
    & \meter{} \\
    \lstick{$\ket{\psi}$}
    & \qw
    & \gate{O_\mc{H}}
    & \push{\cdots}
    & \gate{O_\mc{H}}
    & \gate{O_\mc{H}^{-d/2}}
    & \qw
    \end{quantikz}
    \qquad
    =
    \qquad
    \begin{quantikz}[row sep=0.6cm, column sep=0.45cm]
    \lstick{$\ket{0}$}
    & \gate{e^{\I\varphi_d X}}
    & \ctrl{1} & \octrl{1}
    & \push{\cdots}
    & \ctrl{1} & \octrl{1}
    & \gate{e^{\I\varphi_0 X}}
    & \meter{} \\
    \lstick{$\ket{\psi}$}
    & \qw
    & \gate{O_\mc{H}^{1/2}} & \gate{O_\mc{H}^{-1/2}}
    & \push{\cdots}
    & \gate{O_\mc{H}^{1/2}}  & \gate{O_\mc{H}^{-1/2}}
    & \qw
    & \qw
    \end{quantikz}%
    }
    \caption{Quantum circuits equivalent to the QSP circuit in \cref{fig:qetu}, used in the proof of \cref{thm:qetu_mat_rep}.}
    \label{fig:qetu_proof}
\end{figure}

Because the fractional evolutions commute with other gates, the $d$ copies of $O_\mc{H}^{- 1/2}$ can be redistributed so that each query in the circuit carries a single fractional time step. This redistribution yields the circuit identity shown in \cref{fig:qetu_proof}. Each composite controlled time evolution then has the matrix form
\begin{equation}
    \begin{pmatrix}
        O_\mc{H}^{-1/2} & 0 \\
        0          & O_\mc{H}^{1/2}
    \end{pmatrix}
    = \begin{pmatrix}
        e^{\I \mc{H}/2}   & 0 \\
        0            & e^{-\I \mc{H}/2}
      \end{pmatrix}.
\end{equation}
Let $\ket{\psi_k}$ be an eigenvector of $\mc{H}$ with eigenvalue $\lambda_k$. In the two-dimensional subspace
\begin{equation}
    \mc{B}_{\lambda_k}
    := \mathrm{span}\left\{\ket{0}\ket{\psi_k},\, \ket{1}\ket{\psi_k}\right\},
\end{equation}
this composite controlled evolution acts as a $Z$-rotation $e^{\I \lambda_k Z/2}$. Thus, on each block $\mc{B}_{\lambda_k}$, the circuit reduces to an interleaving of $X$- and $Z$-rotations:
\begin{equation}
    \big[\wt{\mc{U}} (\mc{H},\Phi)\big]_{\mc{B}_{\lambda_k}}
    = e^{\I \varphi_0 X} \prod_{j=1}^d \big( e^{\I \lambda_k Z/2} e^{\I \varphi_j X} \big)
    = \mc{Q}(\lambda_k/2,\Phi).
\end{equation}

Lifting back to the full $2^{n+1}$-dimensional space, we obtain
\begin{equation}
    \wt{\mc{U}} (\mc{H},\Phi)
    = \begin{pmatrix}
        u(\cos(\mc{H}/2)) &
        - \sin(\mc{H}/2)\, w(\cos(\mc{H}/2)) + \I v(\cos(\mc{H}/2)) \\
        \sin(\mc{H}/2)\, w(\cos(\mc{H}/2)) + \I v(\cos(\mc{H}/2)) &
        u(\cos(\mc{H}/2))
      \end{pmatrix}.
\end{equation}
Combining this with \cref{eqn:qetu_circ_rep_relation} gives the full matrix representation of the circuit in \cref{fig:qetu}. In particular, the probabilities of observing ancilla outcomes $\ket{0}$ and $\ket{+}$, starting from $\ket{0}\ket{\psi}$, are
\begin{equation}
    p_{0 \to 0}(\mc{H})
    = \norm{(\bra{0} \otimes I)\, \mc{U} (\mc{H},\Phi)\, \ket{0}\ket{\psi}}_2^2
    = \braket{\psi | u^2(\cos(\mc{H}/2)) | \psi}
\end{equation}
and
\begin{equation}
    p_{0 \to +}(\mc{H})
    = \norm{(\bra{+} \otimes I)\, \mc{U} (\mc{H},\Phi)\, \ket{0}\ket{\psi}}_2^2
    = \frac{1}{2}
      + \braket{\psi | \sin(\mc{H}/2)\, u(\cos(\mc{H}/2))\, w(\cos(\mc{H}/2)) | \psi}.
\end{equation}
This completes the proof.

\end{proof}

%TC:endignore

\end{document}